\documentclass[journal,twoside,web]{ieeecolor}
\usepackage{cite}
\usepackage{generic}
\usepackage{amsmath,amssymb,amsfonts}

\usepackage{amsthm}

\usepackage{algorithmic}
\usepackage{graphicx}
\usepackage{algorithm}
\usepackage{hyperref}
\hypersetup{hidelinks}
\usepackage{textcomp}
\usepackage{subfigure}

\def\BibTeX{{\rm B\kern-.05em{\sc i\kern-.025em b}\kern-.08em
    T\kern-.1667em\lower.7ex\hbox{E}\kern-.125emX}}
\newtheorem{lemma}{Lemma}
\newtheorem{theorem}{Theorem}

\newcommand{\xp}{x_P}
\newcommand{\yp}{y_P}
\newcommand{\xe}{x_E}
\newcommand{\ye}{y_E}
\newcommand{\vpx}{v_{Px}}
\newcommand{\vpy}{v_{Py}}
\newcommand{\ap}{a_P}
\newcommand{\ve}{v_E}
\newcommand{\ttp}{\theta_P}
\newcommand{\tte}{\theta_E}
\newcommand{\vpb}{\bar v_P}
\newcommand{\vp}{v_P}
\newcommand{\ld}{\lambda}
\newcommand{\gm}{\gamma}
\newcommand{\apb}{\bar a_P}
\newcommand{\veb}{\bar v_E}
\newcommand{\cc}{\mathcal{C}}
\newcommand{\bc}{\mathbf{c}}

\newcommand{\bx}{\mathbf{x}}
\newcommand{\tpt}{t_\theta}
\newcommand{\px}{p_x}
\newcommand{\py}{p_y}
\newcommand{\qx}{q_x}
\newcommand{\qy}{q_y}
\newcommand{\lra}{\Leftrightarrow}
\newcommand{\ep}{\epsilon}

\begin{document}
\title{Pursuit-Evasion Between a Velocity-Constrained Double-Integrator Pursuer and a Single-Integrator Evader}
\author{Zehua Zhao, Rui Yan, Jianping He, Xinping Guan, Xiaoming Duan
\thanks{Zehua Zhao, Jianping He, Xinping Guan and Xiaoming Duan are with the State Key Laboratory of Submarine Geoscience, School of Automation and Intelligent Sensing, Shanghai Jiao Tong University, Shanghai 200240, China (e-mails: \{zehua.zhao, jphe, xpguan, xduan\}@sjtu.edu.cn).}
\thanks{Rui Yan is with the School of Artificial Intelligence, Beihang University, Beijing 100191, China (e-mail: rui\_yan@buaa.edu.cn).}
}

\maketitle
\begin{abstract}
We study a pursuit-evasion game  between a double integrator-driven pursuer with bounded velocity and bounded acceleration and a single integrator-driven evader with bounded velocity in a two-dimensional plane. The pursuer's goal is to capture the evader in the shortest time, while the evader attempts to delay the capture. We analyze two scenarios based on whether the capture can happen before the pursuer's speed reaches its maximum. For the case when the pursuer can capture the evader before its speed reaches its maximum, we use geometric methods to obtain the strategies for the pursuer and the evader. For the case when the pursuer cannot capture the evader before its speed reaches its maximum, we use numerical methods to obtain the strategies for the pursuer and the evader. In both cases, we demonstrate that the proposed strategies are optimal in the sense of Nash equilibrium through the Hamilton–Jacobi–Isaacs equation, and the pursuer can capture the evader as long as as its maximum speed is larger than that of the evader. Simulation experiments illustrate the effectiveness of the strategies.
\end{abstract}

\begin{IEEEkeywords}
double-integrator, Hamilton-Jacobi-Isaacs equation, optimal strategies, pursuit-evasion games, velocity constraints.
\end{IEEEkeywords}

\section{Introduction}
With the rapid advancement of autonomy and robotics, pursuit-evasion (PE) games have emerged as an important application for multiagent systems. In such games, pursuers aim to capture evaders as efficiently as possible, while evaders strive to avoid or delay capture. These scenarios are commonly found in natural ecosystems—such as the interaction between predators and prey, and group behaviors \cite{TF-17, YJ-20}—as well as in military applications, including drone tracking, missile interception, and artillery defense \cite{NS-21, YC-18, WI-20}.

The theoretical foundation of PE games traces back to Isaacs’ seminal work in the 1960s, which frames adversarial interactions as differential games and laid the groundwork for modern analysis \cite{IR-99}. Over decades, PE games have evolved into a rich interdisciplinary field, bridging control theory, optimization, and artificial intelligence. Nowadays, based on different objectives, pursuit-evasion games have branched into various problems, such as reach-avoid games \cite{yan2018reach}, perimeter defense problems \cite{shishika2020cooperative}, defense games in a region \cite{yan2017defense}, etc.

In Isaacs' study, to solve differential game problems, it is necessary to solve the Hamilton–Jacobi–Isaacs (HJI) equation, which is a partial differential equation. However, solving the HJI equation is extremely challenging in complex problems. In subsequent research, various methods have been explored to address differential games and pursuit-evasion problems, such as Pontryagin's maximum principle \cite{pontryagin1966theory} and others. Recently, geometric methods have been employed to solve PE games due to their intuitiveness and simplicity \cite{yan2017escape, garcia2020multiple, deng2023multiple, fu2020guarding, ramana2017pursuit, yan2024pursuit}. The approach begins by determining the barrier of the game, which divides the entire game space into different regions based on the advantages of both players. Subsequently, the strategies for both players are derived from this division, and the optimality is verified using the HJI equation \cite{yan2022matching, lee2022two, garcia2017geometric, garcia2018design, G-20}. While solving the HJI equation is challenging, verifying whether the value function satisfies the HJI equation is much easier. This has become a commonly used method for solving PE games.

Despite the various breakthroughs in the previous studies on PE games, such as extending the 2D space to 3D \cite{yan2019construction}, adding a capture radius for the pursuer \cite{yan2022matching}, and extending the one-on-one pursuit-evasion problem to a multi-agent scenario \cite{yan2019task}, the players considered in these problems are mostly driven by single integrators. However, in practical applications, players are often unable to suddenly change both the magnitude and direction of their velocity as in the case of single integrators. To fill this gap, some studies focus on the Dubins model~\cite{yan2024multiplayer, zheng2021time, CZ-19}, but the model is difficult to analyze  due to its nonlinear characteristics. As a result, the problem is often simplified and converted into an optimal control problem by fixing forward speed or choosing stationary targets, which further limits its practical application. 

Another approach is to replace the single integrator-driven players with double integrator ones so that the players' acceleration and turning become smoother, avoiding sudden sharp turns or abrupt acceleration and deceleration. However, due to the geometric complexity of the double integrator model, related research is limited. In \cite{CM-17}, Coon et al. propose a technique for solving pursuit-evasion problems involving double-integrator players using geometric methods: Isochrones. Isochrones are defined as the set of points a player can reach within a certain time under a specific strategy. With the concept of Isochrones, the originally complex geometric properties of pursuit-evasion problems involving double-integrator players are simplified. In \cite{LS-251, LS-252, LS-253}, Li et al. analyze pursuit-evasion problems for three different cases: when the pursuer is a double-integrator, when the evader is a double-integrator, and when both players are double-integrators. They provide the strategies for both players under different initial conditions and ultimately prove the optimality of these strategies using the HJI equation. Although the double-integrator model better aligns with the dynamics of real robots and vehicles, the speed of the player must not increase infinitely. Therefore, limitations need to be applied to ensure that the player's velocity does not become unbounded. One approach is to introduce damping to the acceleration~\cite{LS-251, LS-252, LS-253}, which causes the player's speed to gradually stabilize instead of growing indefinitely. In \cite{lyu2025reach}, Lyu et al. presents a comprehensive study on  this model and adopts it in reach-avoid games. Another method is to impose a hard constraint on the player's velocity, similar to real robots and vehicles that have a rated maximum speed or output saturation, thus ensuring that the player's speed does not exceed a certain threshold. However, imposing a hard constraint on the player's velocity causes the geometric advantages brought by Isochrones to vanish. One can impose additional constraints on the control variables, such that when the velocity approaches the boundary of the constraint, the control variable rapidly increases in the opposite direction, forcing the velocity back into the constrained region~\cite{GW-21}. Or one can use Bang-Off-Bang control, which, according to Pontryagin's Maximum Principle, forces the velocity to reach the constraint boundary by applying the maximum control value, and then sets the control variable to zero, maintaining the velocity at the maximum value \cite{FM-25, JP-05}. However, the problems discussed in \cite{GW-21, FM-25, JP-05} are all one-dimensional, and to our knowledge, there are no articles that apply such a velocity hard-constraint formulation to the pursuit-evasion problem in two-dimensional space. Therefore, finding optimal strategies for a double-integrator pursuit-evasion game with a hard velocity constraint remains an open problem.

In this work, we study the pursuit-evasion game problem in a two-dimensional plane between a double-integrator pursuer~($P$) and a single-integrator evader ($E$). The control input for $P$ consists of the magnitude and direction of acceleration, with constraints on  the maximum acceleration and speed; the control input for $E$ is the magnitude and the direction of speed, also with a constraint on the maximum speed. What's more, $P$ has a hard constraint on its velocity to ensure its speed does not exceed a certain threshold. $P$'s objective is to capture $E$ as quickly as possible, while $E$'s goal is to delay the capture as much as possible. Since $P$'s speed is subject to a hard constraint, our paper develops the optimal strategies under two cases. First, when $P$ can capture $E$ before reaching its maximum speed, there is no speed constraint on $P$, reducing the pursuit-evasion problem to a typical game between a double-integrator $P$ and a single-integrator $E$. Although optimal strategies under various initial conditions have been extensively studied in \cite{LS-251}, the models in these studies involved damping, which can be arbitrarily small but not zero. Therefore, this part of the article complements~\cite{LS-251}, providing a strategy for a model with zero damping and verifying its optimality in the sense of Nash equilibrium using the HJI equation. Second, when $P$ cannot capture $E$ before reaching its maximum speed, Isochrones no longer apply. In this case, the article introduces a simple numerical method to solve for the strategies and uses the HJI equation to verify its optimality in the sense of Nash equilibrium. Our major contributions are as follows. 

\begin{enumerate}
    \item We formulate a PE game involving a double-integrator $P$ with a hard speed constraint and a single-integrator $E$, and we divide the problem into two separate cases: one where $P$ has not yet reached its maximum speed when capture occurs and one where it has. 
    \item In the case when $P$ can capture $E$ before $P$ reaches its maximum speed, we derive the analytical strategies for the PE game using geometric methods. 
    \item In the case when $P$ cannot capture $E$ before $P$ reaches its maximum speed, we propose a novel and feasible numerical method to solve for the strategies.
    \item We verify the optimality of the proposed strategies in the sense of Nash equilibrium using the HJI equation.
\end{enumerate}

The rest of this article is organized as follows. Section \ref{PF} presents the problem fomulation and the HJI equation required for differential games. Section \ref{OS} provides the corresponding strategies for two cases: when $P$ captures $E$ before reaching its maximum speed, and when it does not. The optimality of both strategies in the sense of Nash equilibrium is verified using the HJI equation. We also outline the complete algorithm for computing the optimal strategies. Section \ref{AS} presents the simulation results. Finally, Section \ref{CC} concludes the article.

\section{Problem Formulation}\label{PF}
We consider a pursuer $P$ driven by a double integrator and an evader $E$ driven by a single integrator on a 2D plane, and their dynamics are given by
\begin{equation}
    P:\begin{cases}
        \dot x_P=\vpx,\\\dot y_P=\vpy,\\\dot v_{Px}=\ap\cos\ttp,\\\dot v_{Py}=\ap\sin\ttp,
    \end{cases}
    \quad E:\begin{cases}
        \dot x_E=\ve\cos\tte,\\\dot y_E=\ve\sin\tte,
    \end{cases}
\label{dyn}
\end{equation}\\
where $(\xp,\yp)$ and $(\xe,\ye)$ are the positions of $P$ and $E$, and $(\vpx,\vpy)$ is the velocity of $P$, and $(\xp^0,\yp^0)=(\xp(0),\yp(0))$ and $(\xe^0,\ye^0)=(\xe(0),\ye(0))$ are the initial positions of $P$ and $E$, and $(\vpx^0,\vpy^0)=(\vpx(0),\vpy(0))$ is the initial velocity of $P$. We denote the system state by $\mathbf x=(\mathbf x_P^\top,\mathbf x_E^\top)^\top=(\xp,\yp,\vpx,\vpy,\xe,\ye)^\top$, where $\mathbf x_P=(\xp,\yp,\vpx,\vpy)^\top$ and $\mathbf x_E=(\xe,\ye)^\top$ are states of $P$ and $E$, respectively, and the initial state by $\mathbf x^0=({\mathbf x_P^0}^\top,{\mathbf x_E^0}^\top)^\top=(\xp^0,\yp^0,\vpx^0,\vpy^0,\xe^0,\ye^0)^\top$. The control inputs are the magnitude $\ap$ and the direction $\ttp$ of $P$'s acceleration and the magnitude $\ve$ and the direction $\tte$ of $E$'s velocity. The magnitudes of $P$'s acceleration and $E$'s velocity are assumed to be bounded, i.e., $\ap\in[0,\apb]$, $\ve\in[0,\veb]$. Moreover, to ensure that the speed of $P$ will not increase indefinitely, the magnitude of $P$'s velocity  is also bounded, i.e., $\vp=\sqrt{\vpx^2+\vpy^2}\in[0, \vpb]$. The capture occurs when the positions of $P$ and $E$ coincide, i.e., $\xp=\xe$ and $\yp=\ye$. On the other hand, we also assume the maximum speed of $P$ is bigger than that of $E$, i.e., $\vpb>\veb$, which ensures that the capture can occur (see Lemma \ref{lm5} for details).

In the PE game, $P$ aims to capture $E$ as soon as possible, while $E$ wants to delay the capture, and we define the cost function of the game as
\begin{equation}
    J=t_f=\int_0^{t_f}dt,
\label{cofu}
\end{equation}
where $t_f$ is the capture time. The terminal set is defined by $\{\mathbf x=(\mathbf x_P^\top,\mathbf x_E^\top)^\top\,|\,\Psi(\mathbf x)=0\}$, where
\begin{equation}
    \Psi(\mathbf x)=(\xp-\xe)^2+(\yp-\ye)^2.
\label{psi}
\end{equation}

Since $P$ and $E$ aim to find the optimal strategies to minimize or maximize the cost function in the game, the optimal strategies $\ap^*$, $\ttp^*$, $\ve^*$, $\tte^*$ must satisfy
\begin{equation*}
\begin{aligned}
    J(\mathbf x,\ap^*,\ttp^*,\ve^*,\tte^*)&=\min_{\ap,\ttp}\max_{\ve,\tte}J(\mathbf x,\ap,\ttp,\ve,\tte)\\&=\max_{\ve,\tte}\min_{\ap,\ttp}J(\mathbf x,\ap,\ttp,\ve,\tte).
\end{aligned}
\end{equation*}
This implies that under the optimal strategies, neither $P$ nor $E$ can achieve a better outcome in the game by unilaterally changing their own strategy,  i.e.,
\begin{equation*}
    \begin{aligned}
        J(\mathbf x,\ap^*,\ttp^*,\ve^*,\tte^*)\geq J(\mathbf x,\ap^*,\ttp^*,\ve,\tte),\\J(\mathbf x,\ap^*,\ttp^*,\ve^*,\tte^*)\leq J(\mathbf x,\ap,\ttp,\ve^*,\tte^*),
    \end{aligned}
\end{equation*}
hold for any $\ap$, $\ttp$, $\ve$, and $\tte$. Moreover, the value function of this PE game is given by
\begin{equation}
    V=\min_{\ap,\ttp}\max_{\ve,\tte}J=\max_{\ve,\tte}\min_{\ap,\ttp}J.
\label{valuef}
\end{equation}
According to \cite{IR-99}, the strategies of the PE games are optimal in the sense of Nash equilibrium if and only if the value function $V$ satisfies the following HJI equation
\begin{multline}
        \frac{\partial V}{\partial\xp}\vpx+\frac{\partial V}{\partial\yp}\vpy+\frac{\partial V}{\partial\xe}\ve^*\cos\tte^*+\frac{\partial V}{\partial\ye}\ve^*\sin\tte^*\\+\frac{\partial V}{\partial\vpx}\ap^*\cos\ttp^*+\frac{\partial V}{\partial\vpy}\ap^*\sin\ttp^*+1=0,
\label{hji}
\end{multline}
where $\ap^*$, $\ttp^*$, $\ve^*$, $\tte^*$ are the optimal strategies of $P$ and $E$.

\section{Optimal Strategies}\label{OS}
In this section, we will present strategies for $P$ and $E$ under different initial conditions in Subsection~\ref{subsection:beforemax} and~\ref{subsection:aftermax}. Then, we will provide the algorithm for computing these strategies in Subsection~\ref{Alg}. Finally, we will prove the optimality of these strategies in the sense of Nash equilibrium using the HJI equation~\eqref{hji} in Subsection \ref{subsection:optimal}.

Unlike games where both $P$ and $E$ are driven by single integrators, in our game, $P$ is driven by a double integrator, and simple geometric methods cannot be applied to obtain the strategies. Additionally, a hard constraint is imposed on $P$'s motion by setting an upper bound on its velocity to prevent its speed from increasing indefinitely, and the strategies for $P$ and $E$ depend on whether $P$ can capture $E$ before reaching its maximum speed. In the following, we analyze two cases.
\subsection{Strategies when the pursuer can capture the evader before reaching the maximum speed}\label{subsection:beforemax}
We first study the case when $P$ can capture $E$ before reaching the maximum speed. In this case, the hard constraint on the motion of $P$ is inactive, and we can obtain the following lemma using the Hamiltonian.
\begin{lemma}[Necessary conditions for optimal strategies  when the pursuer can capture the evader before reaching its maximum speed]
    If $P$ can capture $E$ before $P$'s speed reaches the maximum, i.e.,  $\vp(t_f)<\vpb$, then the optimal strategy for $P$ is to accelerate along a fixed direction and maintain the maximum acceleration, i.e., $\ap^*=\apb$ and $\ttp^*$ is constant, while the optimal strategy for $E$ is to move with the maximum speed in a fixed direction, i.e., $\ve^*=\veb$ and $\tte^*$ is constant.
\label{lm1}
\end{lemma}
\begin{proof}
    The Hamiltonian of \eqref{dyn} is
    \begin{equation}
        \begin{aligned}
            H=&\ld_1\vpx+\ld_2\vpy+\ld_3\ap\cos\ttp+\ld_4\ap\sin\ttp\\&+\gm_1\ve\cos\tte+\gm_2\ve\sin\tte+1,
        \end{aligned}
    \label{hmt}
    \end{equation}
    where $\ld_1$, $\ld_2$, $\ld_3$, $\ld_4$, $\gm_1$ and $\gm_2$ are costates. According to the Pontryagin Maximum Principle, we have
    \begin{equation*}
        \begin{aligned}
            \dot\ld_1=-\frac{\partial H}{\partial\xp}=0,\quad\dot\ld_2=-\frac{\partial H}{\partial\yp}=0,\\\dot\gm_1=-\frac{\partial H}{\partial\xe}=0,\quad\dot\gm_2=-\frac{\partial H}{\partial\ye}=0,
        \end{aligned}
    \end{equation*}
    so the costates $\ld_1$, $\ld_2$, $\gm_1$ and $\gm_2$ are constant. Again, according to the Pontryagin Maximum Principle, we have
    \begin{equation*}
        \begin{aligned}
            \dot\ld_3&=-\frac{\partial H}{\partial\vpx}=-\ld_1,\quad\ld_3(t_f)=\mu_1\frac{\partial\Psi}{\partial\vpx}=0,\\\dot\ld_4&=-\frac{\partial H}{\partial\vpy}=-\ld_2,\quad\ld_4(t_f)=\mu_2\frac{\partial\Psi}{\partial\vpy}=0,
        \end{aligned}
    \end{equation*}
    where $\mu_1$ and $\mu_2$ are Lagrange multipliers and $\Psi$ is given by~\eqref{psi}. Therefore, we have
    \begin{equation}
        \ld_3(t)=-\ld_1t+\ld_1t_f,\quad\ld_4(t)=-\ld_2t+\ld_2t_f.
    \label{ldd}
    \end{equation}
    
    $P$ wants the Hamiltonian \eqref{hmt} to be small, while $E$ aims for the opposite. Thus, from the Hamiltonian \eqref{hmt} and \eqref{ldd}, we have
    \begin{equation*}
        \begin{aligned}
            \cos\ttp^*&=-\frac{\ld_3}{\sqrt{\ld_3^2+\ld_4^2}}=-\frac{\ld_1}{\sqrt{\ld_1^2+\ld_2^2}},\\\sin\ttp^*&=-\frac{\ld_4}{\sqrt{\ld_3^2+\ld_4^2}}=-\frac{\ld_2}{\sqrt{\ld_1^2+\ld_2^2}},\\\cos\tte^*&=\frac{\gm_1}{\sqrt{\gm_1^2+\gm_2^2}},\\\sin\tte^*&=\frac{\gm_2}{\sqrt{\gm_1^2+\gm_2^2}},
        \end{aligned}
    \end{equation*}
    which means that $\ttp^*$ and $\tte^*$ are constant.
    
    For $\ap^*$ and $\ve^*$, we have
    \begin{equation*}
        \begin{aligned}
            \frac{\partial H}{\partial\ap}&=\ld_3\cos\ttp+\ld_4\sin\ttp=-\sqrt{\ld_3^2+\ld_4^2}<0,\\\frac{\partial H}{\partial\ve}&=\gm_1\cos\tte+\gm_2\sin\tte=\sqrt{\gm_1^2+\gm_2^2}>0.
        \end{aligned}
    \end{equation*}
    
    For $P$ (or $E$), in order to minimize (or maximize) the Hamiltonian, $\ap$ (or $\ve$) should take the maximum, and thus
    \begin{equation*}
        \ap^*=\apb,\quad\ve^*=\veb.
    \end{equation*}
\end{proof}

From Lemma \ref{lm1}, we know that, if $P$ can capture $E$ before $P$ reaches its maximum speed, the optimal strategy for $P$ is to use the maximum acceleration and to maintain a constant direction of acceleration, while the optimal strategy for $E$ is to move with the maximum velocity in a fixed direction. Using these results, we can obtain the  positions that $P$ and $E$ can reach at a given time $t$ before $P$ reaches its maximum speed.
\begin{lemma}[Reachability circles]
    If $P$ and $E$ move according to the strategies in Lemma \ref{lm1}, then the positions that $P$ and $E$ can reach at time $t$ before $P$ reaches its maximum speed form two circles $\cc_P$ and $\cc_E$, respectively, and the centers and radii of them are
    \begin{align}\label{secc}
        \begin{split}
            \bc_P(t)&=(\xp^0+\vpx^0t,\yp^0+\vpy^0t)^\top,\quad\bc_E=(\xe^0,\ye^0)^\top,\\R_P(t)&=\frac{1}{2}\apb t^2,\quad R_E(t)=\veb t.
        \end{split}
    \end{align}
\label{lm2}
\end{lemma}
\begin{proof}
    From Lemma~\ref{lm1}, we know that the optimal strategy for $P$ is to use the maximum acceleration and to maintain a constant direction of acceleration. Thus, the position that $P$ can reach at time $t$ before reaching its maximum speed when $P$ moves with the maximum acceleration $\apb$ in the direction of $\ttp$ can be described by
    \begin{equation}
        \begin{cases}
            \xp(\ttp,t)=\xp^0+\vpx^0 t+\frac{1}{2}\apb\cos\ttp\cdot t^2,\\\yp(\ttp,t)=\yp^0+\vpy^0 t+\frac{1}{2}\apb\sin\ttp\cdot t^2,
        \end{cases}
    \label{pt}
    \end{equation}
    which can be equivalently rewritten in the form of the standard equation of a circle:
    \begin{equation}
        (\xp-\xp^0-\vpx^0 t)^2+(\yp-\yp^0-\vpy^0 t)^2=\frac{1}{4}\apb^2t^4.
    \label{cp}
    \end{equation}

    Similarly, the position that $E$ can reach at time $t$ before $P$ reaches its maximum speed when $E$ moves with the maximum velocity $\veb$ in the $\tte$ direction is
    \begin{equation}
        \begin{cases}
            \xe(\tte,t)=\xe^0+\veb\cos\tte\cdot t,\\\ye(\tte,t)=\ye^0+\veb\sin\tte\cdot t,
        \end{cases}
    \label{et}
    \end{equation}
    which can be rewritten as
    \begin{equation}
        (\xe-\xe^0)^2+(\ye-\ye^0)^2=\veb^2t^2.
    \label{ce}
    \end{equation}

\end{proof}

From~\eqref{secc}, we notice that as time $t$ progresses, the center of $\cc_P$ moves with a constant velocity that is equal to the initial velocity of $P$, and the radius of $\cc_P$ expands at a rate that is a quadratic function of $t$. Meanwhile, the center of $\cc_E$ remains stationary, and the radius of $\cc_E$  expands at a rate that is a linear function of $t$. Therefore, after a certain period of time, $\cc_E$ must eventually be contained within $\cc_P$. Moreover, during this time period, there must exist a moment when $\cc_E$ is internally tangent to $\cc_P$. By analyzing the process from when $\cc_P$ and $\cc_E$ are disjoint to when $\cc_E$ is contained within $\cc_P$, we obtain the following lemma.

\begin{lemma}[Tangency-based capture guarantee]\label{lm3}
    Suppose $P$ and $E$ move according to the strategies in Lemma \ref{lm1} and $P$ can capture $E$ before $P$'s speed reaches the maximum. If $\cc_E$ is internally tangent to $\cc_P$ at time $t_0$, then $P$ can always capture $E$ no later than $t_0$ regardless of the strategy chosen by $E$.
\end{lemma}
\begin{proof}
    By~\eqref{secc}, the parametric equation of the circle $\cc_E$ is
    \begin{equation*}
        \mathbf x_E(\tte,t)=\mathbf c_E+R_E(t)\cdot\mathbf u_E,
    \end{equation*}
    where $\mathbf u_E=(\cos\tte,\sin\tte)^\top$ is a unit vector. Define a displacement vector $\mathbf V(t)=\mathbf c_P(t)-\mathbf c_E$, whose magnitude is the distance $D(t)$ between the centers of circles $\cc_P$ and $\cc_E$, i.e. $D(t)=\|\mathbf V(t)\|$. For any $\tte\in[0,2\pi)$ chosen by $E$, if $E$ is captured by $P$ at time $t$, then the position of $P$ at this moment, denoted by $T$, must lie on $\cc_E$. The coordinate $\mathbf x_T=\mathbf c_E+R_E(t)\cdot\mathbf u_E$ of $T$ must satisfy:
    \begin{equation*}
        \begin{aligned}
            &\|\mathbf c_P(t)-\mathbf x_T\|=R_P(t)\\
            \Leftrightarrow&\|\mathbf V(t)-R_E(t)\mathbf u_E\|=R_P(t)\\
            \lra&\|\mathbf V(t)\|^2-2R_E(t)\mathbf u^\top_E\mathbf V(t)+R_E^2(t)\|\mathbf u_E\|^2=R_P^2(t)\\
            \lra&D^2(t)-2R_E(t)\mathbf u^\top_E\mathbf V(t)+R_E^2(t)=R_P^2(t)\\
            \lra&\mathbf u^\top_E\mathbf V(t)=\frac{D^2(t)+R_E^2(t)-R_P^2(t)}{2R_E(t)}.
        \end{aligned}
    \end{equation*}
    Let
    \begin{align*}
        S(t)&=\frac{D^2(t)+R_E^2(t)-R_P^2(t)}{2R_E(t)},\\
        g(t,\tte)&=\mathbf u^\top_E\mathbf V(t)-S(t).
    \end{align*}
    Then $E$ is captured by $P$ at time $t$ when $E$ moves in the $\tte$ direction if and only if $g(t,\tte)=0$.

    When the game has progressed for a short period of time $t_\epsilon$, the circles $\cc_P$ and $\cc_E$ are disjoint, and $\mathbf x_T$ satisfies:
    \begin{equation*}
        \|\mathbf c_P(t_\ep)-\mathbf x_T\|>R_P(t_\ep),
    \end{equation*}
    which is equivalent to $g(t_\ep,\tte)<0$.

    When $t=t_0$, $\cc_E$ is internally tangent to $\cc_P$. If $P$ has not captured $E$ before this moment, $\mathbf x_T$ satisfies:
    \begin{equation*}
    \|\mathbf c_P(t_0)-\mathbf x_T\|\le R_P(t_0),
    \end{equation*}
    which is $g(t_0,\tte)\ge0$.

    Since $g(t, \tte)$ is continuous with respect to $t$, and for any $\tte \in [0, 2\pi)$, we have $g(t_\ep, \tte) < 0$ and $g(t_0, \tte) \ge 0$. By the Intermediate Value Theorem, there exists $t' \in (t_\ep, t_0]$ such that $g(t', \tte) = 0$, and in this case $E$ is captured by $P$ at time $t'$ when $E$ moves in the $\tte$ direction.
\end{proof}

To obtain the strategies for $P$ and $E$ when $P$ can capture $E$ before $P$’s speed reaches the maximum, we need to compute the time it takes for $P$ to reach its maximum speed for different acceleration directions.
\begin{lemma}[Time when the pursuer reaches the max speed]When $P$ follows the strategy given in Lemma~\ref{lm1} and selects $\ttp$ as the direction of acceleration, the time required for $P$ to reach the maximum speed is given by:
\begin{equation}
    \begin{aligned}
            \tpt(\ttp)=&\frac{\sqrt{\vpb^2-(\vpx^0\sin\ttp-\vpy^0\cos\ttp)^2}}{\apb}\\&-\frac{\vpx^0\cos\ttp+\vpy^0\sin\ttp}{\apb}.
    \end{aligned}  
\label{tptp}
\end{equation}
\label{lmnew}
\end{lemma}
\begin{proof}
    By Lemma~\ref{lm1}, $P$ accelerates with the  maximum acceleration before reaching its maximum speed. Therefore, the velocity components of $P$ along the $x$- and $y$-axes satisfy:
    \begin{equation}
        \begin{cases}
            \vpx(\ttp,t)=\vpx^0+\apb\cos\ttp\cdot t,\\\vpy(\ttp,t)=\vpy^0+\apb\sin\ttp\cdot t.
        \end{cases}
    \label{vpt}
    \end{equation}
    When $P$ reaches its maximum speed at time $\tpt(\ttp)$, we have $\vpx^2(\ttp,\tpt(\ttp))+\vpy^2(\ttp,\tpt(\ttp))=\vpb^2$. Combining with~\eqref{vpt}, we obtain~\eqref{tptp}.
\end{proof}

To ensure that $P$ can capture $E$ before reaching its maximum speed when $P$ chooses $\ttp$ as its acceleration direction, it is necessary that $t_f < \tpt(\ttp)$.

We now attempt to derive the strategies for the case when $P$ can capture $E$ before reaching its maximum speed. Since $E$'s goal is to maximize the capture time during the time interval $(0,t_0]$, where $t_0$ is defined in Lemma~\ref{lm3} as the time when $\cc_E$ is internally tangent to $\cc_P$, $E$ should choose a strategy such that it is captured by $P$ at time $t_0$. In other words, $E$ aims to delay capture by $P$ until $\cc_E$ is internally tangent to $\cc_P$ as illustrated in Fig.~\ref{fig1}.
\begin{figure}[!t]
\centerline{\includegraphics[width=\columnwidth]{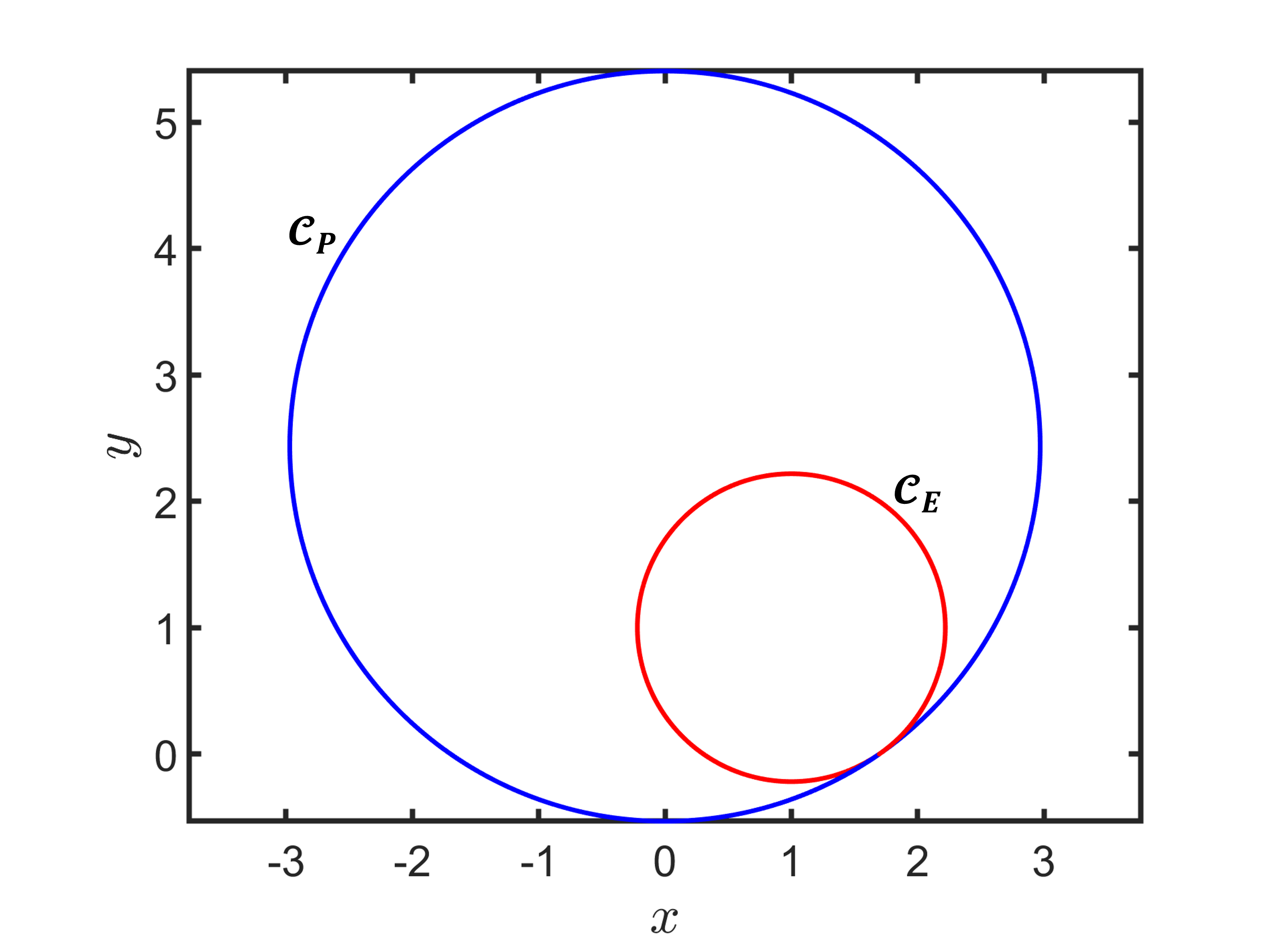}}
\caption{The circle $\cc_E$ is internally tangent to the circle $\cc_P$.}
\label{fig1}
\end{figure}
When $\cc_E$ is internally tangent to $\cc_P$, the distance between $\bc_P$ and $\bc_E$ is equal to the difference in the radii of $\cc_P$ and $\cc_E$. Using this property and~\eqref{secc}, we can obtain
\begin{equation}\label{queq}
    \Gamma(t)=0,
\end{equation}
where 
\begin{multline*}
  \Gamma(t)=\\  (\xp^0-\xe^0+\vpx^0t)^2+(\yp^0-\ye^0+\vpy^0t)^2-(\frac{1}{2}\apb t^2-\veb t)^2
\end{multline*}
is a quartic equation in $t$. All positive solutions of~\eqref{queq} correspond to the moments when $\cc_P$ and $\cc_E$ are internally tangent. Note that there are two cases when \eqref{queq}  holds true: $\cc_P$ is inscribed within $\cc_E$ or $\cc_E$ is inscribed within $\cc_P$. To ensure that the capture time $t$ corresponds to the case when $\cc_E$ is inscribed within $\cc_P$, the radius of $\cc_P$ must be greater than or equal to that of $\cc_E$, i.e. $R_P(t)\geq R_E(t)$. By~\eqref{secc}, we obtain $t\geq\frac{2\veb}{\apb}$. We define the set $\mathcal{T}$ as the set  of all positive $t$ that satisfy $t\geq\frac{2\veb}{\apb}$ and \eqref{queq}, then we have $t_f \in \mathcal{T}$.

For any $t\in \mathcal{T}$, we can determine the  equations for circles $\cc_P$ and $\cc_E$ at time $t$ using \eqref{cp} and \eqref{ce}, and further compute the coordinates of the tangency point $\bx_f=(x_f,y_f)^\top$ as
\begin{equation}
        \begin{aligned}
            x_f=\frac{\apb t^2\xe^0-2\veb t(\xp^0+\vpx^0t)}{\apb t^2-2\veb t},\\y_f=\frac{\apb t^2\ye^0-2\veb t(\xp^0+\vpy^0t)}{\apb t^2-2\veb t},
        \end{aligned}
    \label{tanp}
\end{equation}
which is the capture point corresponding to $t\in \mathcal{T}$. Furthermore, by \eqref{pt} and \eqref{et}, where $\bx_P(\ttp,t)=\bx_E(\ttp,t)=\bx_f$, we obtain $\ttp$ and $\tte$ for $P$ and $E$ as follows
\begin{equation}
        \begin{aligned}
            \cos\ttp&=\cos\tte=\frac{\xp^0+\vpx^0t-\xe^0}{\veb t-\frac{1}{2}\apb t^2},\\\sin\ttp&=\sin\tte=\frac{\yp^0+\vpy^0t-\ye^0}{\veb t-\frac{1}{2}\apb t^2}.
        \end{aligned}
    \label{ops0}
\end{equation}

Since we consider the case when $P$ can capture $E$ before reaching its maximum speed, the capture time $t_f$ must satisfy $t_f < \tpt(\ttp^*)$, where $\ttp^*$ is the acceleration direction of $P$ corresponding to the capture time $t_f$ and can be calculated by \eqref{ops0}. Moreover, for each $t\in \mathcal{T}$, $\cc_E$ is internally tangent to $\cc_P$, and by Lemma~\ref{lm3}, $P$ is guaranteed to capture $E$ on or before time $t$. Since $P$'s goal is to capture $E$ as quickly as possible, if multiple instances occur during the game in which $\cc_E$ is inscribed in $\cc_P$, $P$ should choose to execute the capture at the first such instance. In other words, the capture time $t_f$ should be the smallest number in $\mathcal{T}$ that satisfies  $t_f<\tpt(\ttp^*)$. The above discussion provides a method for determining the capture time $t_f$ in the case when $P$ is able to capture $E$ before reaching its maximum speed. With the capture time $t_f$, we can compute the coordinates of the tangency point by \eqref{tanp}, and further obtain $\ttp^*$ and $\tte^*$ for $P$ and $E$ by~\eqref{ops0}. Finally, according to Lemma \ref{lm1}, we can give the strategies for $P$ and $E$ under the capture time $t_f$ when capture can occur before $P$ reaches its maximum speed by~\eqref{ops0} as
\begin{equation}
    \begin{aligned}
            \cos\ttp^*&=\cos\tte^*=\frac{\xp^0+\vpx^0t_f-\xe^0}{\veb t_f-\frac{1}{2}\apb t_f^2},\\\sin\ttp^*&=\sin\tte^*=\frac{\yp^0+\vpy^0t_f-\ye^0}{\veb t_f-\frac{1}{2}\apb t_f^2},\\
            \ap^*&=\apb,\quad\quad\ve^*=\veb.
        \end{aligned}
    \label{ops1}
\end{equation}
Since the strategies in \eqref{ops1} are in a closed form with respect to the capture time $t_f$, and $t_f$ is an analytical solution to the quartic equation \eqref{queq}, the strategies in \eqref{ops1} are analytical.

\subsection{Strategies when the pursuer cannot capture the evader before reaching the maximum speed}\label{subsection:aftermax}
In this subsection, we  discuss the strategies for $P$ and $E$ when $P$ cannot capture $E$ before $P$ reaches its maximum speed. We emphasize that the derived results are not direct extensions of strategies in \eqref{ops1} for the case when $P$ can capture $E$ before reaching the maximum speed. Instead, entirely new strategies are developed that account for the whole game process, from the initial game state to the capture event.

\begin{lemma}[Necessary conditions for optimal strategies  when the pursuer cannot capture the evader before reaching its maximum speed]
    If $P$ cannot capture $E$ before $P$ reaches its maximum speed, i.e., $\vp(t_f)=\vpb$, then the optimal strategy for $P$ has two phases: i) before $P$ reaches its maximum speed, $P$'s optimal strategy is to maintain a fixed acceleration direction and accelerate at the maximum rate until the maximum speed is reached, i.e., when $\vp<\vpb$, $\ap^*=\apb$ and $\ttp^*$ is constant; ii) afterwards, the acceleration becomes zero, and $P$ continues to move at maximum speed along the direction of velocity at the moment it reaches the maximum speed, i.e., when $\vp=\vpb$, $\ap^*=0$. Moreover, $E$'s optimal strategy is to move at the maximum speed along a fixed direction, i.e., $\ve^*=\veb$ and $\tte^*$ is constant.
\label{lm4}
\end{lemma}
\begin{proof}
    Since $P$ cannot capture $E$ before $P$ reaches its maximum speed, there exists the state constraint
    \begin{equation}
        G(\mathbf x)=\vpx^2+\vpy^2-\vpb^2\le0,
    \label{sc}
    \end{equation}
    and the Hamiltonian of \eqref{dyn} is
    \begin{equation}
        \begin{aligned}H=&\ld_1\vpx+\ld_2\vpy+\ld_3\ap\cos\ttp+\ld_4\ap\sin\ttp\\&+\gm_1\ve\cos\tte+\gm_2\ve\sin\tte+1\\&+\nu(\vpx^2+\vpy^2-\vpb^2),
        \end{aligned}
        \label{zhekezaban}
    \end{equation}
    where $\ld_1$, $\ld_2$, $\ld_3$, $\ld_4$, $\gm_1$ and $\gm_2$ are costates, and $\nu\ge0$ is the Lagrange multiplier associated with the state constraint \eqref{sc}. By the Pontryagin Maximum Principle, we have
    \begin{equation*}
        \begin{aligned}
            \dot\gm_1=-\frac{\partial H}{\partial\xe}=0,\quad\dot\gm_2=-\frac{\partial H}{\partial\ye}=0,
        \end{aligned}
    \end{equation*}
    so the costates $\gm_1$ and $\gm_2$ are constant. $E$ wants the Hamiltonian \eqref{zhekezaban} to be as large as possible. Therefore, from the Hamiltonian \eqref{zhekezaban} we obtain
    \begin{equation*}
        \begin{aligned}
            \cos\tte^*=-\frac{\gm_1}{\sqrt{\gm_1^2+\gm_2^2}},\quad\sin\tte^*=-\frac{\gm_2}{\sqrt{\gm_1^2+\gm_2^2}},
        \end{aligned}
    \end{equation*}
    which implies that $\tte^*$ is constant.

    For $\ve^*$, we have
    \begin{equation*}
        \begin{aligned}
            \frac{\partial H}{\partial\ve}=\gm_1\cos\tte+\gm_2\sin\tte=\sqrt{\gm_1^2+\gm_2^2}>0.
        \end{aligned}
    \end{equation*}
    Thus, for $E$ to maximize the Hamiltonian \eqref{zhekezaban}, $\ve$ should take the maximum speed, i.e., $\ve^*=\veb$.
    
    By the Karush-Kuhn-Tucker conditions, we have
    \begin{equation*}
        \nu(\vpx^2+\vpy^2-\vpb^2)=0.
    \end{equation*}
    Therefore, there are two possible cases: i) $G(\mathbf x)<0$ and $\nu=0$; or ii) $G(\mathbf x)=0$ and $\nu\ge0$. These two cases correspond to situations where $P$ has not yet reached its maximum speed and where $P$ has already reached its maximum speed, respectively.
    
    The first case is $G(\mathbf x)<0$ and $\nu=0$. In this case, $P$ has not yet reached its maximum speed, and the state constraint~\eqref{sc} is inactive. By the proof of Lemma~\ref{lm1}, we know that  $\ttp^*$ is constant and $\ap^*=\apb$.

    The second case is $G(\mathbf x)=0$ and $\nu\ge0$. In this case, $P$ has already reached its maximum speed, and the state constraint~\eqref{sc} is active, i.e.,
    \begin{equation}
        G(\mathbf x)=\vpx^2+\vpy^2-\vpb^2=0.
    \label{sca}
    \end{equation}
    Taking the derivative of both sides of~\eqref{sca} with respect to time $t$, we obtain:
    \begin{equation*}
        \ap(\vpx\cos\ttp+\vpy\sin\ttp)=0.
    \end{equation*}
    For the above equation to hold, either $\ap = 0$ or $\vpx\cos\ttp+\vpy\sin\ttp=0$, meaning that the acceleration is zero and the velocity direction remains fixed, or the acceleration direction is perpendicular to the velocity direction. In the following, we show that  setting the acceleration to zero allows $P$ to capture $E$ more quickly. 
    
    Let $t_c$ denote the moment when $P$ reaches its maximum speed. For $t \in [t_c, t_f]$, the motion of $E$ satisfies
    \begin{equation}
            \begin{cases}
                \xe(t)=\xe(t_c)+\veb\cos\tte^*\cdot(t-t_c),\\\ye(t)=\ye(t_c)+\veb\sin\tte^*\cdot(t-t_c).
            \end{cases}
    \label{shenmee}
    \end{equation}
    Let the velocity direction angle of $P$ be $\delta(t)$, then we have
    \begin{equation*}
        \begin{cases}
            \vpx(t)=\vpb\cos\delta(t),\\\vpy(t)=\vpb\sin\delta(t),
        \end{cases}
    \end{equation*}
    and the motion of $P$ satisfies
    \begin{equation}
            \begin{cases}
                \xp(t)=\xp(t_c)+\vpb\int_{t_c}^t\cos\delta(\tau)d\tau,\\\yp(t)=\yp(t_c)+\vpb\int_{t_c}^t\sin\delta(\tau)d\tau.
            \end{cases}
    \label{shenmep}
    \end{equation}
    Since $\xp(t_f)=\xe(t_f)$ and $\yp(t_f)=\ye(t_f)$ at the moment of capture, substituting these terminal conditions into the motion equations \eqref{shenmee} and \eqref{shenmep}, we obtain
    \begin{equation*}
        \begin{aligned}
            \vpb\int_{t_0}^
            {t_f}\cos\delta(t)dt&=\veb\cos\tte^*\cdot(t_f-t_c)+\xe(t_c)-\xp(t_c),\\\vpb\int_{t_c}^
            {t_f}\sin\delta(t)dt&=\veb\sin\tte^*\cdot(t_f-t_c)+\ye(t_c)-\yp(t_c).
        \end{aligned}
    \end{equation*}
    Let $\mathbf A = (\int_{t_c}^{t_f}\cos\delta(t)dt,\int_{t_c}^{t_f}\sin\delta(t)dt)^\top$ and $\mathbf B = \frac{1}{\vpb}(\veb\cos\tte^*\cdot(t_f-t_c)+\xe(t_c)-\xp(t_c),\veb\sin\tte^*\cdot(t_f-t_c)+\ye(t_c)-\yp(t_c))^\top$, then we can rewrite the capture condition as $\mathbf{A} = \mathbf{B}$, i.e., $\mathbf{A}$ and $\mathbf{B}$ have the same magnitude and direction. Let $\tau = t_f - t_c$, then the objective of $P$ is to  minimize $\tau$. By the triangle inequality for vector-valued function integrals, we have
    \begin{equation}
        \begin{aligned}
            \|\mathbf A\|&=\|(\int_{t_c}^{t_f}\cos\delta(t)dt,\int_{t_c}^{t_f}\sin\delta(t)dt)^\top\|\\&\le\int_{t_c}^{t_f}\|(\cos\delta(t),\sin\delta(t))^\top\|dt\\&=\int_{t_c}^{t_f}1dt=\tau,
        \end{aligned}
    \label{budeng}
    \end{equation}
    and the equality in~\eqref{budeng} holds if and only if $\delta(t)$ is constant, i.e., the velocity direction of $P$ remains fixed.

    Define the function
    \begin{equation*}
        \begin{aligned}
             f(\tau)=\|\mathbf B\|=&\frac{1}{\vpb}\|(\veb\tau\cos\tte^*+\xe(t_c)-\xp(t_c),\\&\qquad \veb\tau\sin\tte^*+\ye(t_c)-\yp(t_c))^\top\|.
        \end{aligned}
    \end{equation*}
    Then, when capture occurs, the following condition must hold:
    \begin{equation}
        f(\tau)=\|\mathbf B\|=\|\mathbf A\|.
    \end{equation}
    When the velocity of $P$ remains constant, capture necessarily requires that $f(\tau)=\|\mathbf A\|=\tau$, and when the acceleration direction of $P$ is perpendicular to its velocity direction, capture necessarily requires that $f(\tau)=\|\mathbf A\|<\tau$. We next show the first case corresponds to a smaller $\tau$ for $P$.
    
    Define $g(\tau) = f(\tau) - \tau$. Then, we have $g(0) = \frac{1}{\vpb}\|(\xe(t_c)-\xp(t_c),\ye(t_c)-\yp(t_c))^\top\| > 0$, and when $\tau$ is sufficiently large, $g(\tau) = \frac{\sqrt{(\veb\tau\cos\tte^*+\xe(t_c)-\xp(t_c))^2+(\veb\tau\sin\tte^*+\ye(t_c)-\yp(t_c))^2}}{\vpb}-\tau< 0$. Since $g(\tau)$ is continuous on $[0, +\infty)$, by the Intermediate Value Theorem, $g(\tau) = 0$ has at least one solution in $[0, +\infty)$. For any capture moment $\tau_0$ corresponding to the case when the acceleration direction of $P$ is perpendicular to its velocity direction, we have $g(\tau_0) < 0$. Given that $g(0) > 0$ and $g(\tau)$ is continuous, there must exist some $0 < \tau' < \tau_0$ such that $g(\tau') = 0$, which corresponds to the case when $P$ moves with a constant velocity direction. Therefore, any non-straight motion results in a capture time $\tau_0$ strictly greater than a straight-line case $\tau'$. As a result, to minimize the capture time, $P$ should continue moving at the maximum speed along the velocity direction after reaching the maximum speed.
\end{proof}

Next, similar to Lemma~\ref{lm2}, we characterize the positions that $P$ and $E$ can reach at time $t$ under their optimal strategies given by Lemma \ref{lm4}. Note that the position $E$ can reach at time $t$ when moving with a velocity in the direction of $\tte$ is still given by \eqref{et}, so the points that $E$ can reach also form a circle with the same standard equation as $\cc_E$ in \eqref{ce}. Under the strategy given in Lemma \ref{lm4}, $P$ first moves with maximum acceleration until it reaches its maximum speed, and then continues to move along the velocity direction with the maximum speed. According to \eqref{tptp}, given $\ttp$ and for $t\le\tpt(\ttp)$, the positions $P$ can reach at time $t$ when moving under the strategy described in Lemma~\ref{lm4} are still characterized by~\eqref{pt}. While for $t>\tpt(\ttp)$, the positions that $P$ can reach are characterized by
\begin{equation}
    \left\{
        \begin{aligned}
            \xp'(\ttp,t)=\xp^0+\vpx^0 t+\frac{1}{2}\apb\cos\ttp\cdot(2t\tpt(\ttp)-\tpt^2(\ttp)),\\\yp'(\ttp,t)=\yp^0+\vpy^0 t+\frac{1}{2}\apb\sin\ttp\cdot(2t\tpt(\ttp)-\tpt^2(\ttp)).
        \end{aligned}
    \right.
\label{op'}
\end{equation}
From~\eqref{tptp}, we observe that the time required for $P$ to reach its maximum speed varies with the chosen acceleration direction $\ttp$. Therefore, for any given time $t$, $P$ may reach its maximum speed for some acceleration directions,  while for other directions it may not. At this moment, the set of points that $P$ can reach is composed by piecing together~\eqref{pt} and~\eqref{op'}. Regardless of whether the set of points that $P$ can reach is described solely by \eqref{op'}, or jointly by equations \eqref{pt} and \eqref{op'}, we note that under the strategies in Lemma~\ref{lm4}, the positions that $P$ can reach no longer form a circle, but rather form an oval shape shown in Fig.~\ref{fig2}. Nevertheless, we can still derive an important lemma using \eqref{op'}.

\begin{figure}[!t]
\centerline{\includegraphics[width=\columnwidth]{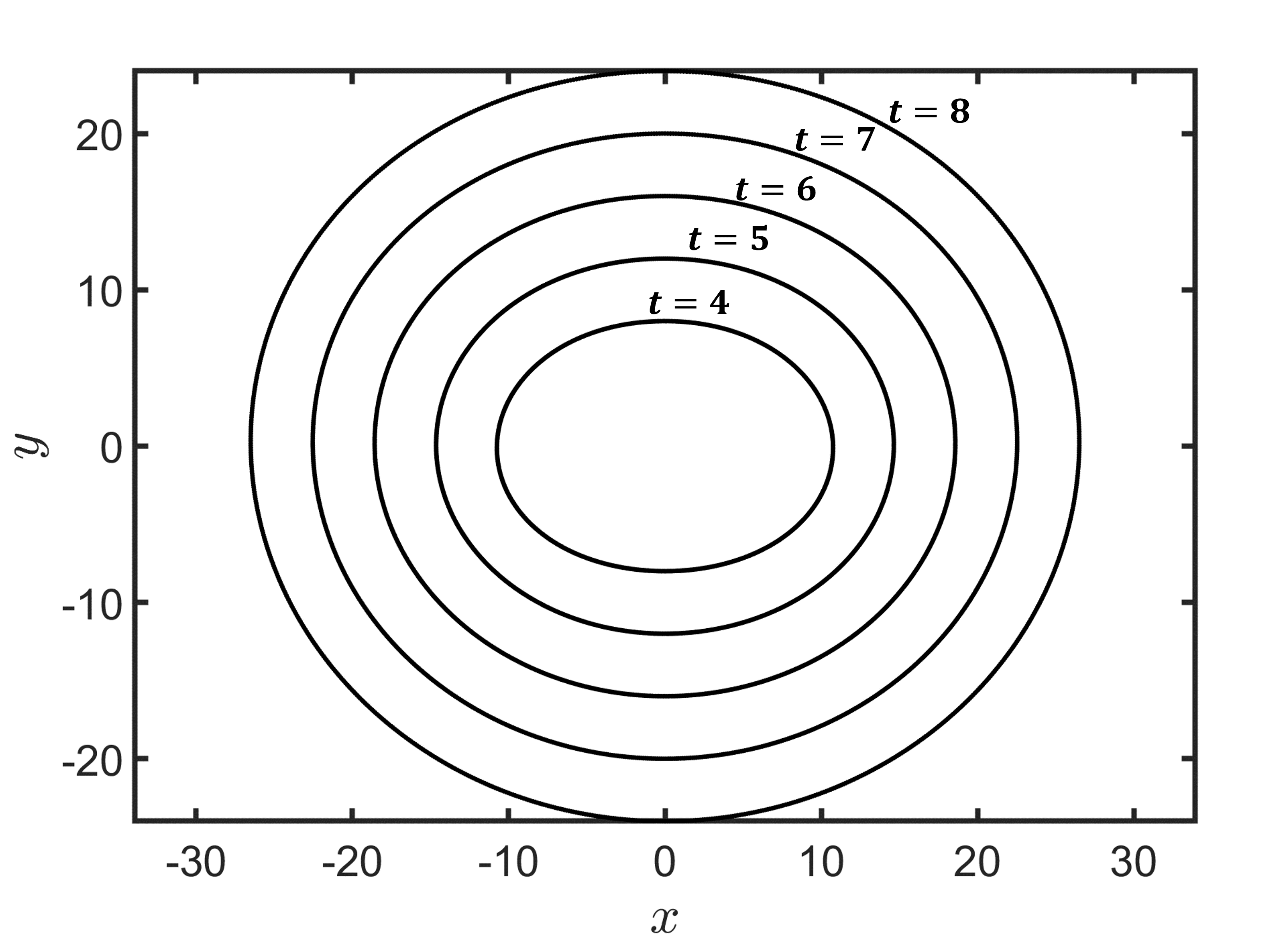}}
\caption{$\xp=0$, $\yp=0$, $\vpx=0$, $\vpy=4$, $\vpb=4$, $\apb=2$, the set of points that P can reach when $t=4$, $t=5$, $t=6$, $t=7$ and $t=8$.}
\label{fig2}
\end{figure}

\begin{lemma}[Capture guarantee with faster pursuer]
    If $\vpb>\veb$, then $P$ is guaranteed to capture $E$.
\label{lm5}
\end{lemma}
\begin{proof} 
    From \eqref{tptp}, we can compute the minimum value of $\tpt(\ttp)$ as $\frac{\vpb-\sqrt{(\vpx^0)^2+(\vpy^0)^2}}{\apb}$, which is obtained when the acceleration direction $\ttp$ is the same as the initial velocity direction of $P$, i.e., $(\cos\ttp,\sin\ttp) = (\vpx^0,\vpy^0)/\|(\vpx^0,\vpy^0)\|$, and the maximum value of $\tpt(\ttp)$ as $\frac{\vpb+\sqrt{(\vpx^0)^2+(\vpy^0)^2}}{\apb}$, which is obtained  when the acceleration direction $\ttp$ is opposite to the initial velocity direction of $P$, i.e., $(\cos\ttp,\sin\ttp) = -(\vpx^0,\vpy^0)/\|(\vpx^0,\vpy^0)\|$. Therefore, when $t > \frac{\vpb+\sqrt{(\vpx^0)^2+(\vpy^0)^2}}{\apb}$, $P$ reaches its maximum speed and begins to move at a constant velocity, regardless of the acceleration direction.
    
    Define $\mathbf c_P'=(\xp^0+\frac{\vpb\vpx^0}{\apb},\yp^0+\frac{\vpb\vpy^0}{\apb})^\top$. When $t > \frac{\vpb+\sqrt{(\vpx^0)^2+(\vpy^0)^2}}{\apb}$, the square of the distance between a point whose position is characterized by \eqref{op'} and $\mathbf c_P'$ is
    \begin{equation}
        \begin{aligned}
             d_P^2(\ttp,t)
            &=((\vpx^0)^2+(\vpy^0)^2)(t-\frac{\vpb}{\apb})^2\\
            &\quad +\frac{1}{2}(2t-\tpt(\ttp))(t-\frac{\vpb}{\apb})\\&\qquad\cdot(\vpb^2-(\vpx^0)^2-(\vpy^0)^2-\apb^2\tpt^2(\ttp))\\&\quad +\frac{1}{4}\apb^2(2t\tpt(\ttp)-\tpt^2(\ttp))^2,
        \end{aligned}
        \label{ddd}
    \end{equation}
    where we used \eqref{tptp} to obtain
    \begin{multline*}
\vpx^0\cos\ttp+\vpy^0\sin\ttp=\\\frac{\vpb^2-(\vpx^0)^2-(\vpy^0)^2-\apb^2\tpt^2(\ttp)}{2\apb\tpt(\ttp)}.
    \end{multline*}
    Note that in~\eqref{ddd}, all terms related to $\ttp$ are represented through $\tpt(\ttp)$, and  \eqref{ddd} can be viewed as a function of $\tpt(\ttp)$ and $t$. When $t$ is fixed, we can find that the minimum value of $d_P(\ttp,t)$ is obtained when $\tpt(\ttp)$ obtains its maximum or minimum value by calculating the derivative of $d_P(\ttp,t)$. Substituting $\tpt(\ttp)=\frac{\vpb-\sqrt{(\vpx^0)^2+(\vpy^0)^2}}{\apb}$ and $\tpt(\ttp)=\frac{\vpb+\sqrt{(\vpx^0)^2+(\vpy^0)^2}}{\apb}$ into \eqref{ddd}, respectively, we can find that the corresponding $d_P(\ttp,t)$ are both equal to $\vpb t-\frac{(\vpx^0)^2+(\vpy^0)^2+\vpb^2}{2\apb}$. Thus when we consider $\mathbf c_P'$ as the center of the shape formed by the points $P$ can reach, the length of the shortest semi-axis of this shape is $d_{P\min}=\vpb t-\frac{(\vpx^0)^2+(\vpy^0)^2+\vpb^2}{2\apb}$, which grows with the rate $\vpb$. From Lemma~\ref{lm2}, we know that $\cc_E$ is centered at $(\xe,\ye)^\top$ with a radius of $\veb t$. Therefore, when $t$ is sufficiently large, the circle $\cc_E$ will be fully contained within the shape formed by the points $P$ can reach. At this time, regardless of $E$'s position on the circle $\cc_E$, $P$ is guaranteed to arrive at that position no later than $t$, and therefore $P$ is guaranteed to capture $E$.
\end{proof}

Next we propose the strategies for the case when $P$ cannot capture $E$ before reaching its maximum speed. Since the capture must occur at the intersection of $\cc_P'$ and $\cc_E$, we can combine \eqref{op'} and \eqref{et} to obtain the coordinates of the capture point. Specifically, the capture time satifies
    \begin{equation}
        (\vpb^2-\veb^2)t^2+2(\px\qx+\py\qy)t+\qx^2+\qy^2=0,
    \label{t2t}
\end{equation}
where
\begin{equation}
        \begin{aligned}
        \px&=\vpx^0+\apb\cos\ttp\cdot\tpt(\ttp),\\\py&=\vpy^0+\apb\sin\ttp\cdot\tpt(\ttp),\\\qx&=\xp^0-\xe^0-\frac{1}{2}\apb\cos\ttp\cdot\tpt^2(\ttp),\\\qy&=\yp^0-\ye^0-\frac{1}{2}\apb\sin\ttp\cdot\tpt^2(\ttp),
        \end{aligned}
    \label{pq}
\end{equation}
and $\px$ and $\py$ satisfy
\begin{equation}
    \px^2+\py^2=\vpb^2.
    \label{eq:relationofp}
\end{equation}
Note that \eqref{t2t} is a quadratic function of $t$, and every coefficient is a function of $\ttp$. Therefore, by solving~\eqref{t2t}, we obtain a formula for $t$ in terms of $\ttp$ as
    \begin{equation}
        t=\frac{g(\bx^0,\ttp)-h(\bx^0,\ttp)}{\vpb^2-\veb^2},
    \label{tpm}
    \end{equation}
in which
\begin{equation}
        \begin{aligned}
        h(\bx^0,\ttp)&=\px\qx+\py\qy,\\g(\bx^0,\ttp)&=\sqrt{h^2(\bx^0,\ttp)-(\vpb^2-\veb^2)(\qx^2+\qy^2)}.
        \end{aligned}
    \label{gh}
\end{equation}
In other words, once the acceleration direction $\ttp$ of $P$ is fixed, the capture time $t_f$ under this strategy is determined. Moreover, since $E$'s objective is to delay capture as much as possible, to determine the strategies for $P$ and $E$, we must first find the $\ttp$ that maximizes $t$ in~\eqref{tpm}. Then we obtain the optimization problem
\begin{equation}
        \max_{\ttp}\quad\frac{g(\bx^0,\ttp)-h(\bx^0,\ttp)}{\vpb^2-\veb^2},
    \label{oppr}
\end{equation}
and the optimal solution $\ttp^*$ and the optimal value $t_f$ of the optimization problem \eqref{oppr} correspond to the acceleration direction of $P$  and the capture time, respectively.

Note that the objective function in the optimization problem~\eqref{oppr} is a periodic function of $\ttp$ with a period of $2\pi$, and the definition of $g(\mathbf x^0,\ttp)$ in \eqref{gh} implicitly requires $w(\mathbf x^0,\ttp)\ge0$, where
\begin{equation}
    w(\mathbf x^0,\ttp)=h^2(\bx^0,\ttp)-(\vpb^2-\veb^2)(\qx^2+\qy^2).
\label{eqw}
\end{equation}
Although a rigorous proof is not available yet, we conjecture that the objective function in~\eqref{oppr} is unimodal over a connected domain in one period, and the optimal solution can be obtained using the ternary search algorithm. To do so, we need to determine the feasible range of $\ttp$ in~\eqref{oppr}, i.e., search for two zeros of $w(\mathbf x^0,\ttp)$ within one period. 
Specifically, we first find $\ttp^-$ and $\ttp^+$ on the interval $[0,2\pi]$ such that $w(\mathbf x^0,\ttp^-) < 0$ and $w(\mathbf x^0,\ttp^+) > 0$, respectively. Starting from $\ttp = 0$, we calculate the function value $w(\mathbf x^0,\ttp)$ at uniformly spaced values with a fixed step size $h = \pi/500$. As soon as a $\theta_{P0}$ is found such that the corresponding function value satisfies $w(\mathbf x^0,0)\cdot w(\mathbf x^0,\theta_{P0})<0$, the search terminates and $\theta_{P0}$ is returned. If no such value is found over the entire interval, the step size is reduced by a factor of $10$, and the process is repeated. This iteration continues until a $\theta_{P0}$ is found such that $w(\mathbf x^0,0)\cdot w(\mathbf x^0,\theta_{P0})<0$. Then, $\ttp=0$ and $\ttp=\theta_{P0}$ correspond to two evaluations of $w(\mathbf x^0,\ttp)$ with opposite signs. Denote the one at which the function value is negative by $\ttp^-$, and the one at which the function value is positive by $\ttp^+$. Suppose $\ttp^-<\ttp^+$ (or $\ttp^->\ttp^+$), then we can use the bisection method to obtain the two zeros $\theta_{P1}$ and $\theta_{P2}$ of $w(\mathbf x^0,\ttp)$ over the intervals $[\ttp^-, \ttp^+]$ (or $[\ttp^--2\pi, \ttp^+]$) and $[\ttp^+, \ttp^- + 2\pi]$ (or $[\ttp^+, \ttp^- ]$), respectively. Finally, we apply the ternary search algorithm over the domain $[\theta_{P1}, \theta_{P2}]$. In our simulations in Section~\ref{AS}, we employ the above procedure to solve~\eqref{oppr}.

By solving~\eqref{oppr}, we obtain the optimal solution $\ttp^*$ and the optimal value $t_f$ of the optimization problem \eqref{oppr}, which correspond to the acceleration direction of $P$  and the capture time, respectively. By substituting $\ttp^*$ into \eqref{op'}, we obtain the capture point $\bx_f=(x_f,y_f)^\top$ as
\begin{equation}
        \begin{aligned}
            x_f&=\xp^0+\vpx^0t_f+\frac{1}{2}\apb\cos\ttp^*(2t_f\tpt(\ttp^*)-\tpt^2(\ttp^*)),\\y_f&=\yp^0+\vpy^0t_f+\frac{1}{2}\apb\sin\ttp^*(2t_f\tpt(\ttp^*)-\tpt^2(\ttp^*)).
        \end{aligned}
    \label{xff}
\end{equation}
According to Lemma~\ref{lm4}, $E$ needs to move at maximum speed in a fixed direction towards the capture point. Therefore, the strategies of $P$ and $E$ for the case when $P$ cannot capture $E$ before reaching its maximum speed are
    \begin{equation}
        \begin{aligned}
            \cos\tte^*&=\frac{x_f-\xe^0}{\sqrt{(x_f-\xe^0)^2+(y_f-\ye^0)^2}},\\\sin\tte^*&=\frac{y_f-\ye^0}{\sqrt{(x_f-\xe^0)^2+(y_f-\ye^0)^2}},\\\ap^*&=\left\{
            \begin{aligned}
                \apb\quad\vp<\vpb,\\0\quad\vp=\vpb,
            \end{aligned}
            \right.\\\ve^*&=\veb,
        \end{aligned}
    \label{opst2}
    \end{equation}
    and $\ttp^*$ is the optimal solution to~\eqref{oppr}.

\subsection{Algorithm}\label{Alg}

So far we have presented the strategies for $P$ and $E$ when capture can occur both before and after $P$ reaches its maximum speed. Therefore, in order to derive strategies under different initial conditions, we need to determine whether $P$ can capture $E$ before reaching its maximum speed.

In Lemma \ref{lmnew}, we provide the formula for $\tpt(\ttp)$ and its physical meaning, which will serve as the condition for determining whether $P$ can capture $E$ before reaching its maximum speed. Under the strategie \eqref{ops1}, based on the current states of $P$ and $E$, we can compute the capture time $t_f$ and the acceleration direction $\ttp^*$ for $P$. By substituting $\ttp^*$ into  \eqref{tptp}, we can obtain the time $\tpt(\ttp^*)$ for $P$ to reach the maximum speed in the current acceleration direction $\ttp^*$. Then, we compare $\tpt(\ttp^*)$ with the capture time $t_f$. If $t_f\leq\tpt(\ttp^*)$, then $P$ can capture $E$ before reaching its maximum speed, and thus the strategies in~\eqref{ops1} are valid. If otherwise $t_f>\tpt(\ttp^*)$, then $P$ cannot capture $E$ before reaching its maximum speed, and the strategies in~\eqref{ops1} are not valid. In this case, $P$ and $E$ must apply the strategies in~\eqref{opst2}.

With the strategies and the condition for determining their validity, we present the entire algorithm for the PE game. In Algorithm~\ref{alg1}, the inputs are the current state variables $\mathbf x^0$ of both $P$ and $E$, as well as their respective constraints $\apb$, $\vpb$, and $\veb$. 
First, we check whether $P$ is able to capture $E$ before reaching its maximum speed. We compute the set of solutions $\mathcal T$ of \eqref{queq} that satisfy $t\ge\frac{2\vpb}{\apb}$ on Line \ref{alsolveT}. If $\mathcal T$ is nonempty, then we select the smallest $t \in \mathcal T$ as the provisional capture time $t_f$ on Line \ref{algettf}. Then, using \eqref{ops1}, we determine the corresponding $\ttp^*$ and compute $\tpt(\ttp^*)$ by \eqref{tptp} on Line \ref{algetst1} and \ref{algetttheta}. If $t_f \leq \tpt(\ttp^*)$, then $P$ can capture $E$ before reaching its maximum speed and  $t_f$ is the capture time. In this case, the strategies $\ap^*$, $\ttp^*$, $\ve^*$, and $\tte^*$ for $P$ and $E$ at the current state $\mathbf x^0$ are given by~\eqref{ops1}. If $t_f > \tpt(\ttp^*)$, then $t_f$ does not satisfy the condition for capture before $P$ reaching its maximum speed, and it must be removed from the set $\mathcal T$ on Line \ref{alremove}; we then repeat the above procedure with the next smallest element $t$ of $\mathcal T$. When $\mathcal T$ is empty, which indicates that $P$ cannot capture $E$ before reaching its maximum speed, we must apply \eqref{opst2} to derive the strategies for the current state $\mathbf x^0$ on Line~\ref{alst21} and~\ref{alst22}.

\begin{algorithm}[h]
\caption{Solving for the Strategies.}\label{alg:alg1}
\begin{algorithmic}[1]
\REQUIRE $\mathbf x^0,\apb,\vpb,\veb$
\ENSURE $\ap^*,\ttp^*,\ve^*,\tte^*$
\STATE solve \eqref{queq} and obtain $\mathcal T=\{t>\frac{2\veb}{\apb}\,|\,\Gamma(t)=0\}$\label{alsolveT}
\WHILE{$\mathcal T\neq\emptyset$}
\STATE $t_f=\min T$\label{algettf}
\STATE calculate $\ap^*,\ttp^*,\ve^*,\tte^*$ by \eqref{ops1}\label{algetst1}
\STATE calculate $\tpt(\ttp^*)$ by \eqref{tptp}\label{algetttheta}
\IF{$t_f\le\tpt(\ttp^*)$}
\RETURN $\ap^*,\ttp^*,\ve^*,\tte^*$ at current time
\ELSE
\STATE $\mathcal T=\mathcal T\setminus\{t_f\}$\label{alremove}
\ENDIF
\ENDWHILE
\STATE obtain $\ttp^*$ and $t_f$ by solving \eqref{oppr}\label{alst21}
\STATE calculate $\ap^*,\ttp^*,\ve^*,\tte^*$ by \eqref{opst2}\label{alst22}
\RETURN $\ap^*,\ttp^*,\ve^*,\tte^*$ at current time
\end{algorithmic}
\label{alg1}
\end{algorithm}

\subsection{Optimality of the strategies}\label{subsection:optimal}
We have already obtained the strategies for $P$ and $E$ no matter $P$ can capture $E$ before $P$ reaches its maximum speed or not, but we still need to verify the optimality of these strategies in the sense of Nash equilibrium using the HJI equation~\eqref{hji}.
\begin{theorem}[Optimality of strategies in the sense of Nash equilibrium]
\label{thm3}
    The value function \eqref{valuef} satisfies HJI equation~\eqref{hji}, which means the strategies \eqref{ops1} and \eqref{opst2} for $P$ and $E$ of this PE game are optimal in the sense of Nash equilibrium.
\end{theorem}
\begin{proof}
    The proof is postponed to the appendix.
\end{proof}


\section{Simulation}\label{AS}

\begin{figure*}[!t]
\centering
\subfigure[$P$ and $E$ both use the optimal strategies.]{\includegraphics[width=2.3in]{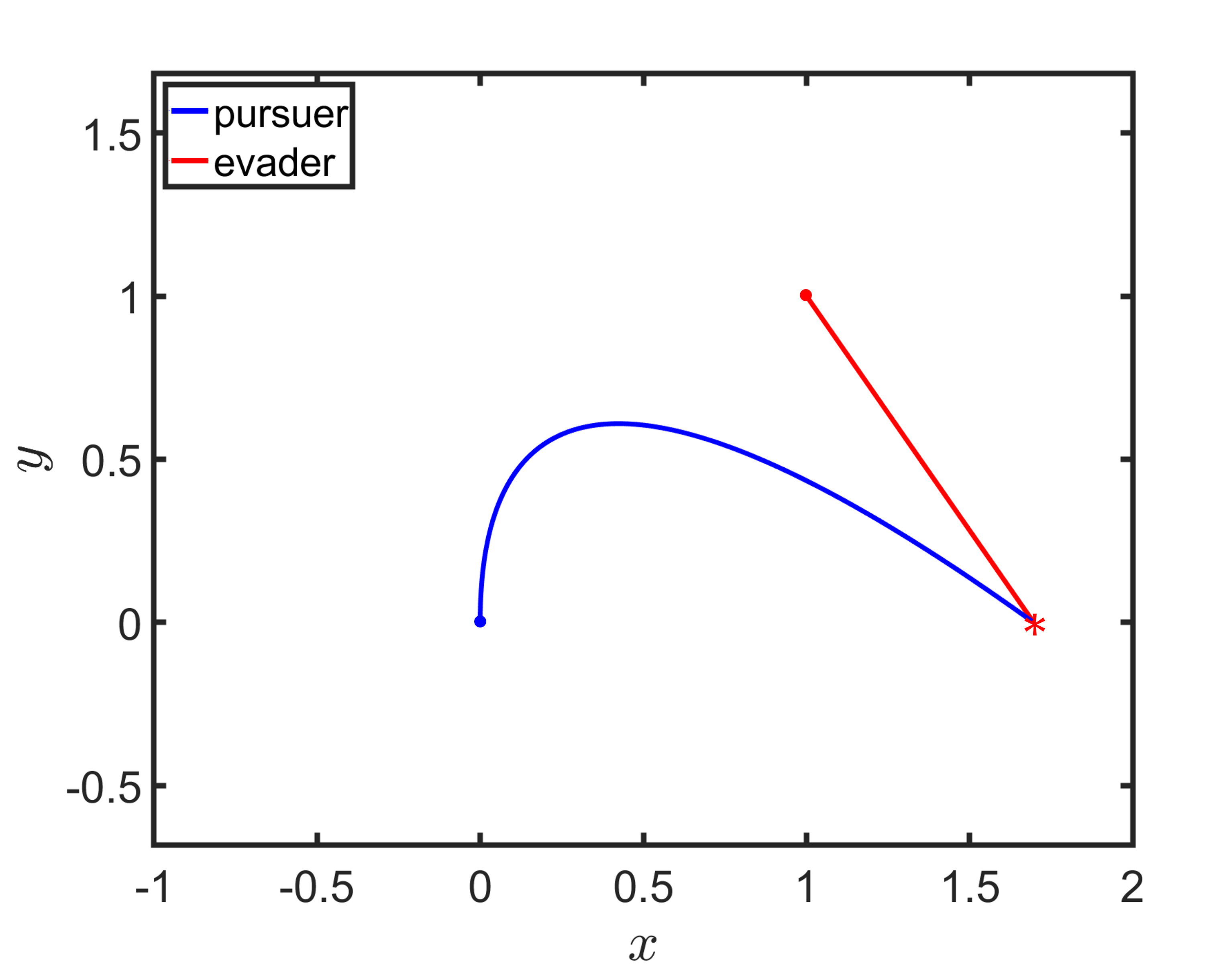}}
\subfigure[$P$ uses the pure-pursuit strategy while $E$ uses the optimal strategy.]{\includegraphics[width=2.3in]{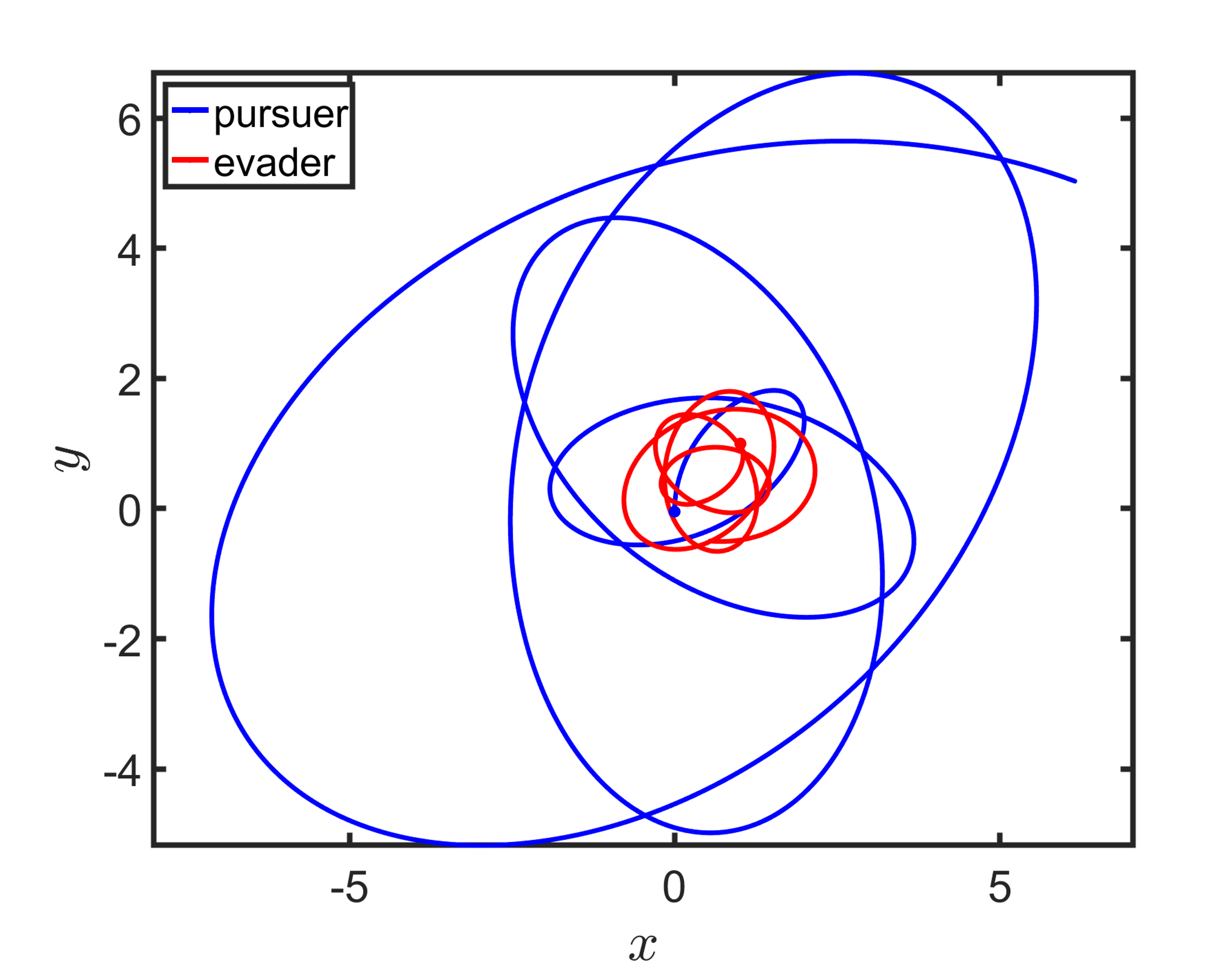}}
\subfigure[$P$ uses the optimal strategy while $E$ uses the pure-evasion strategy.]{\includegraphics[width=2.3in]{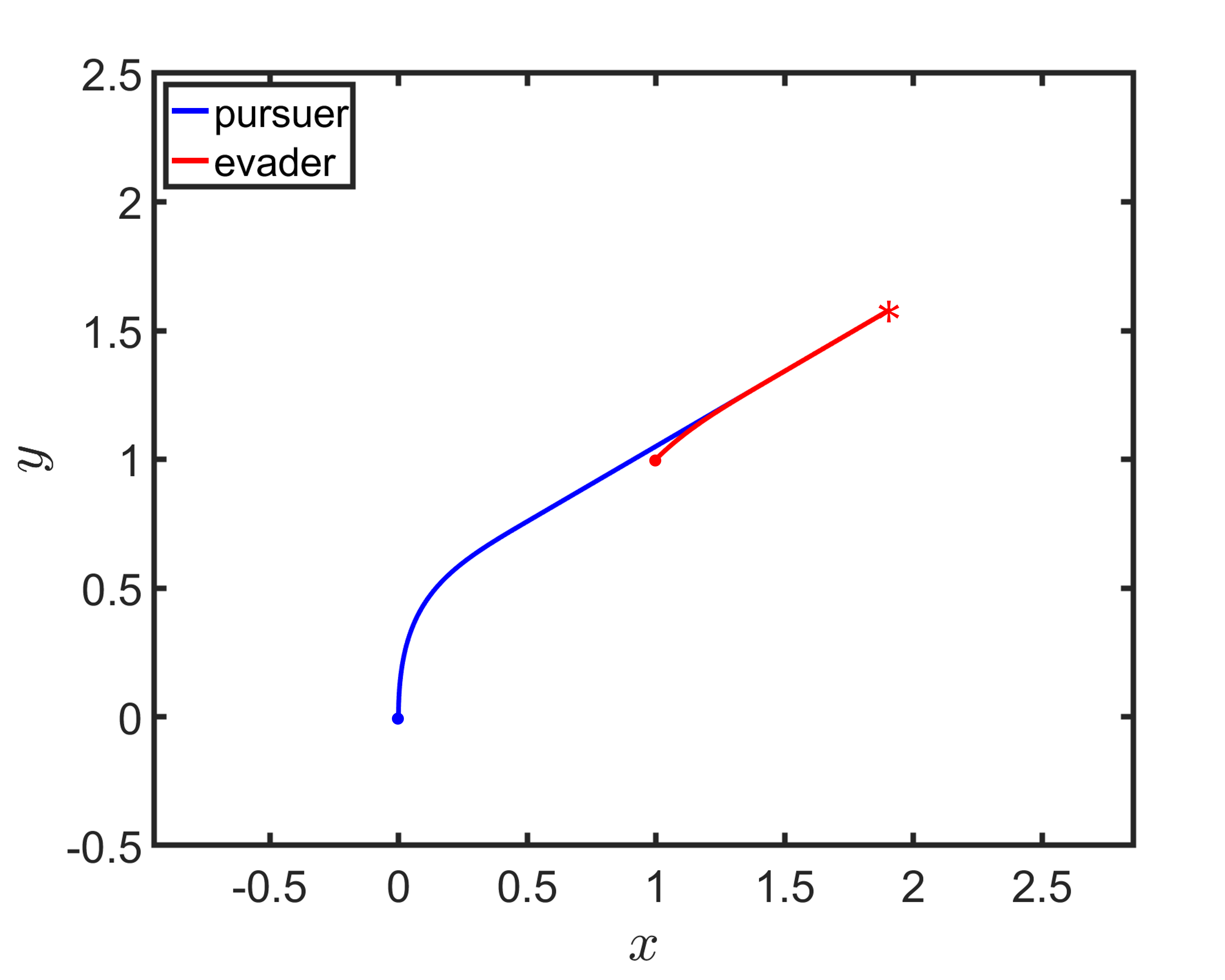}}
\caption{Scenario I: $\mathbf x^0=(0,0,0,1,1,1)^\top$, $\vpb=10$, $\apb=1$, $\veb=0.5$.}
\label{fig3}
\end{figure*}
\begin{figure*}[!t]
\centering
\subfigure[$P$ and $E$ both use the optimal strategies.]{\includegraphics[width=2.3in]{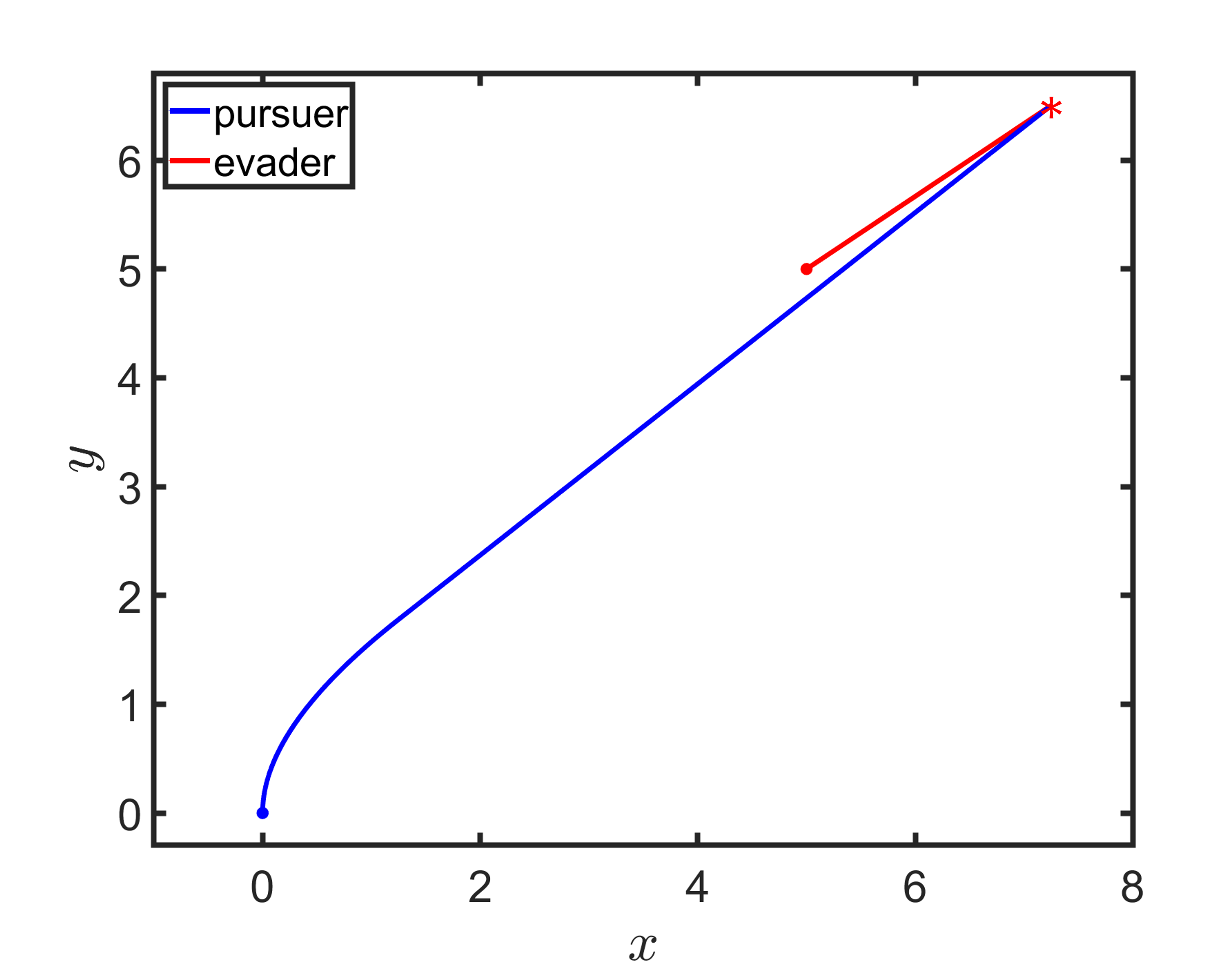}}
\subfigure[$P$ uses the pure-pursuit strategy while $E$ uses the optimal strategy.]{\includegraphics[width=2.3in]{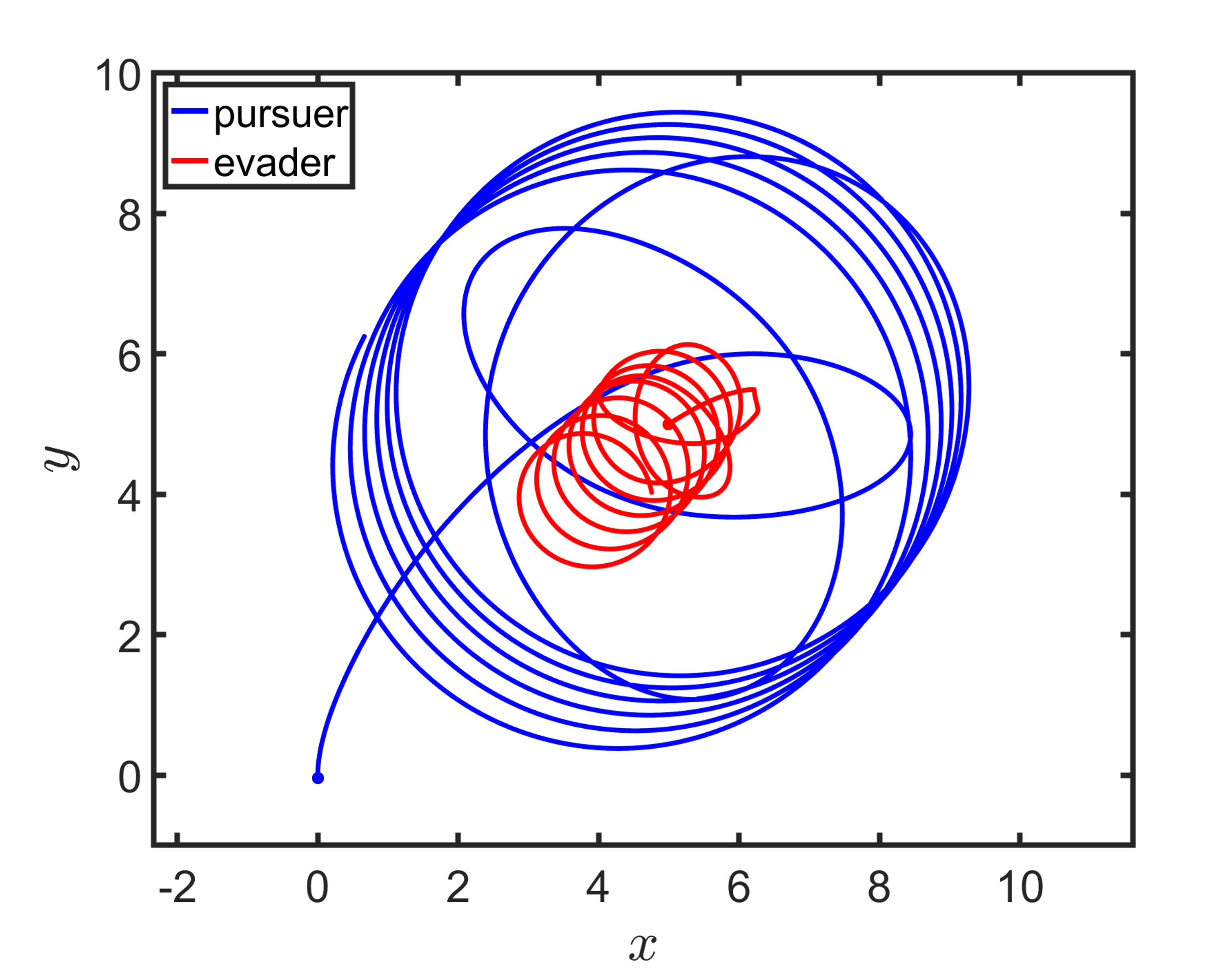}}
\subfigure[$P$ uses the optimal strategy while $E$ uses the pure-evasion strategy.]{\includegraphics[width=2.3in]{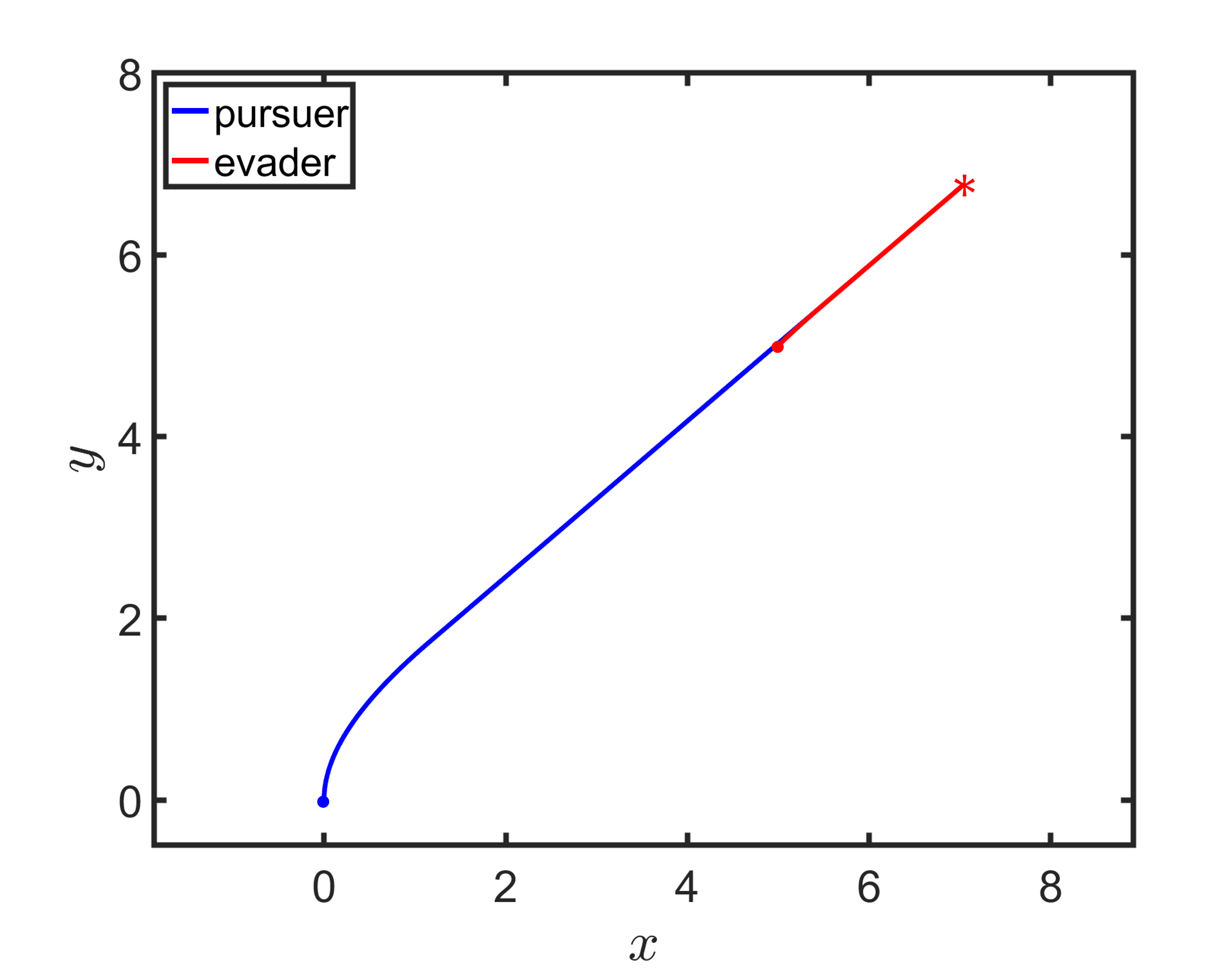}}
\caption{Scenario II: $\mathbf x^0=(0,0,0,1,5,5)^\top$, $\vpb=2$, $\apb=1$, $\veb=0.5$.}
\label{fig4}
\end{figure*}

In this section, we present some simulations to illustrate the effectiveness of our proposed strategies. All simulations are produced using MATLAB R2023b. The hardware configuration is as follows: CPU: 13th Gen Intel® Core™ i9-13980HX @ 2.20 GHz, Memory: 16.0 GB RAM.

Since we have proposed two different strategies based on whether $P$ can capture $E$ before reaching its maximum speed, we provide two distinct simulation scenarios corresponding to these two strategies. In Scenario I, the initial state $\mathbf x^0=(0,0,0,1,1,1)^\top$, and $\vpb=10$, $\apb=1$, $\veb=0.5$, where $P$ can capture $E$ before reaching its maximum speed under the optimal strategies. In Scenario II, the initial state $\mathbf x^0=(0,0,0,1,5,5)^\top$, and $\vpb=2$, $\apb=1$, $\veb=0.5$, where $P$ cannot capture $E$ before reaching its maximum speed under the optimal strategies. The simulation results of the optimal strategies in these two scenarios are shown in Fig. \ref{fig3}(a) and Fig. \ref{fig4}(a), respectively.

To illustrate the advantages of our proposed strategies, we adopt the pure-pursuit strategy and the pure-evasion strategy for comparison. When $P$ uses the pure-pursuit strategy, its acceleration direction always points toward $E$'s current position. When $E$ uses the pure-evasion strategy, its velocity direction is always the same as the line starting from $P$ and pointing to $E$. The simulation results when $P$ uses the pure-pursuit strategy while $E$ uses the optimal strategy and when $P$ uses the optimal strategy while $E$ uses the pure-evasion strategy in these two scenarios are shown in Fig.~\ref{fig3}(b), Fig.~\ref{fig3}(c) and Fig.~\ref{fig4}(b), Fig.~\ref{fig4}(c), respectively. The capture times when $P$ and $E$ use different strategies are reported in Table \ref{table}, which validates that the proposed strategies perform better.

\begin{table}
\caption{Capture Times}
\label{table}
\setlength{\tabcolsep}{3pt}
\begin{tabular}{|p{70pt}|p{70pt}|p{45pt}|p{45pt}|}
\hline
$P$'s Strategy& 
$E$'s Strategy& 
Capture Time In Scenario I& 
Capture Time In Scenario II\\
\hline
the optimal strategy& 
the optimal strategy& 
2.437& 
5.407\\
\hline
the optimal strategy& 
the pure-evasion \par strategy& 
2.155& 
5.397\\
\hline
the pure-pursuit \par strategy& 
the optimal strategy& 
$+\infty$& 
$+\infty$\\
\hline
\end{tabular}
\label{tab1}
\end{table}

\section{Conclusion}\label{CC}

We study a pursuit-evasion game between a double integrator-driven pursuer and a single integrator-driven evader, where the pursuer has a constraint on the magnitude of its velocity. If the pursuer is able to capture the evader before reaching its maximum speed, then the optimal strategy for the pursuer is to apply maximum acceleration along a  fixed direction, while the evader moves in a fixed direction at maximum speed, and both players move toward the capture point. And we provide specific strategies for the purser and the evader using geometric methods. If the pursuer cannot capture the evader before reaching its maximum speed, then the optimal strategy for the pursuer is to accelerate with the maximum acceleration along a fixed direction until reaching the maximum speed, and then continues moving at this speed in the same direction, while the evader moves in a fixed direction at maximum speed, and both players move toward the capture point. The capture point can be solved using numerical optimization methods. The optimality of these strategies in the sense of Nash equilibrium is verified using the HJI equation. Simulation results show that the proposed strategies are indeed the optimal strategies in the sense of Nash equilibrium. The strategies provide a feasible solution to pursuit-evasion problems in complex real-world scenarios such as drone tracking and autonomous driving. Future research could further extend this work to three-dimensional space or multi-agent collaborative scenarios.

\appendix

\section*{Proof of Theorem \ref{thm3}}
In this PE game, we utilize two different strategies under different initial conditions, depending on whether $P$ can capture $E$ before reaching its maximum speed. To prove that the strategies for this game are optimal in the sense of Nash equilibrium, we need to demonstrate that the value function \eqref{valuef} under the strategies satisfies the HJI equation~\eqref{hji}. 
Then, we must also demonstrate that when the initial conditions change continuously, leading to a switch in strategies, the value function \eqref{valuef} remains continuous. We note that in order to establish the optimality of a strategy in the sense of Nash equilibrium, the HJI equation~\eqref{hji} must hold for all states. Therefore, in the following proof, the initial state $\mathbf{x}^0$ will be replaced by a generic state $\mathbf{x}$ at any time. 

    First, we demonstrate the optimality of strategies \eqref{ops1} in the sense of Nash equilibrium, where the value function~\eqref{valuef} is given by $V=t_f$. To verify the HJI equation~\eqref{hji}, we need the partial derivatives of $V$ with respect to each state variable. Since $V=t_f$ is the solution to \eqref{queq}, we perform implicit differentiation on both sides of \eqref{queq} and obtain
    \begin{equation}
        \begin{aligned}
            \frac{\partial V}{\partial\xp}&=\frac{d_x}{D},\quad\frac{\partial V}{\partial\yp}=\frac{d_y}{D},\\\frac{\partial V}{\partial\xe}&=-\frac{d_x}{D},\quad\frac{\partial V}{\partial\ye}=-\frac{d_y}{D},\\\frac{\partial V}{\partial\vpx}&=\frac{d_xt_f}{D},\quad\frac{\partial V}{\partial\vpy}=\frac{d_yt_f}{D},
        \end{aligned}
        \label{h1}
    \end{equation}
    where
    \begin{equation}
        \begin{aligned}
            d_x&=\xp-\xe+\vpx t_f,\\d_y&=\yp-\ye+\vpy t_f,\\D&=-d_x\vpx-d_y\vpy+(\frac{1}{2}\apb t_f^2-\veb t_f)(\apb t_f-\veb).
        \end{aligned}
    \label{h2}
    \end{equation}
    Notice that
    \begin{equation}
        \begin{aligned}
            &d_x\cos\ttp^*+d_y\sin\ttp^*\\=&\frac{(\xp-\xe+\vpx t_f)^2+(\yp-\ye+\vpy t_f)^2}{\veb t_f-\frac{1}{2}\apb t_f^2}\\=&\frac{(\frac{1}{2}\apb t_f^2-\veb t_f)^2}{\veb t_f-\frac{1}{2}\apb t_f^2}\\=&\veb t_f-\frac{1}{2}\apb t_f^2,
        \end{aligned}
    \label{h3}
    \end{equation}
    where we used \eqref{queq} in the second equality. Substituting~\eqref{h1}, \eqref{h2}, and \eqref{h3} into the HJI equation~\eqref{hji}, we obtain
    \begin{equation*}
        \begin{aligned}
            &\frac{\partial V}{\partial\xp}\vpx+\frac{\partial V}{\partial\yp}\vpy+\frac{\partial V}{\partial\xe}\ve^*\cos\tte^*+\frac{\partial V}{\partial\ye}\ve^*\sin\tte^*\\&+\frac{\partial V}{\partial\vpx}\ap^*\cos\ttp^*+\frac{\partial V}{\partial\vpy}\ap^*\sin\ttp^*+1\\=&\frac{d_x\vpx+d_y\vpy-d_x\veb\cos\tte^*-d_y\veb\sin\tte^*}{D}\\&+\frac{d_xt_f\apb\cos\ttp^*+d_yt_f\apb\sin\ttp^*}{D}+1\\=&\frac{d_x\vpx+d_y\vpy+(d_x\cos\ttp^*+d_y\sin\ttp^*)(\apb t_f-\veb)}{D}+1\\=&\frac{d_x\vpx+d_y\vpy-(\frac{1}{2}\apb t_f^2-\veb t_f)(\apb t_f-\veb)}{D}+1=0,
        \end{aligned}
    \end{equation*}
    where we used the strategies given in~\eqref{ops1}. 
    Thus the value function $V=t_f$ satisfies HJI equation~\eqref{hji}, which means the strategies in~\eqref{ops1} are optimal in the sense of Nash equilibrium.

    In the following, we demonstrate the optimality of strategies~\eqref{opst2} in the sense of Nash equilibrium. 
    According to \eqref{tpm}, we know that $t_f$ depends on $\px$, $\py$, $\qx$ and $\qy$ defined in~\eqref{pq}. Therefore, we first compute the partial derivatives of them with respect to each state variable as follows. For $\px$ we have
    \begin{equation}
        \begin{aligned}
            \frac{\partial\px}{\partial\xp}&=\frac{\partial\px}{\partial\ttp^*}\cdot\frac{\partial\ttp^*}{\partial\xp},\quad\frac{\partial\px}{\partial\yp}=\frac{\partial\px}{\partial\ttp^*}\cdot\frac{\partial\ttp^*}{\partial\yp},\\
            \frac{\partial\px}{\partial\xe}&=\frac{\partial\px}{\partial\ttp^*}\cdot\frac{\partial\ttp^*}{\partial\xe},\quad\frac{\partial\px}{\partial\ye}=\frac{\partial\px}{\partial\ttp^*}\cdot\frac{\partial\ttp^*}{\partial\ye},\\
            \frac{\partial\px}{\partial\vpx}&=\frac{\partial\px}{\partial\ttp^*}\cdot\frac{\partial\ttp^*}{\partial\vpx}+R_1\cos\ttp^*+1,\\
            \frac{\partial\px}{\partial\vpy}&=\frac{\partial\px}{\partial\ttp^*}\cdot\frac{\partial\ttp^*}{\partial\vpy}+R_2\cos\ttp^*,
        \end{aligned}
    \label{pxpr}
    \end{equation}
    for $\py$ we have
    \begin{equation}
        \begin{aligned}
            \frac{\partial\py}{\partial\xp}&=\frac{\partial\py}{\partial\ttp^*}\cdot\frac{\partial\ttp^*}{\partial\xp},\quad\frac{\partial\py}{\partial\yp}=\frac{\partial\py}{\partial\ttp^*}\cdot\frac{\partial\ttp^*}{\partial\yp},\\
            \frac{\partial\py}{\partial\xe}&=\frac{\partial\py}{\partial\ttp^*}\cdot\frac{\partial\ttp^*}{\partial\xe},\quad\frac{\partial\py}{\partial\ye}=\frac{\partial\py}{\partial\ttp^*}\cdot\frac{\partial\ttp^*}{\partial\ye},\\
            \frac{\partial\py}{\partial\vpx}&=\frac{\partial\py}{\partial\ttp^*}\cdot\frac{\partial\ttp^*}{\partial\vpx}+R_1\sin\ttp^*,\\
            \frac{\partial\py}{\partial\vpy}&=\frac{\partial\py}{\partial\ttp^*}\cdot\frac{\partial\ttp^*}{\partial\vpy}+R_2\sin\ttp^*+1,
        \end{aligned}
    \label{pypr}
    \end{equation}
    for $\qx$ we have
    \begin{equation}
        \begin{aligned}
            \frac{\partial\qx}{\partial\xp}&=\frac{\partial\qx}{\partial\ttp^*}\cdot\frac{\partial\ttp^*}{\partial\xp}+1,\quad\frac{\partial\qx}{\partial\yp}=\frac{\partial\qx}{\partial\ttp^*}\cdot\frac{\partial\ttp^*}{\partial\yp},\\
            \frac{\partial\qx}{\partial\xe}&=\frac{\partial\qx}{\partial\ttp^*}\cdot\frac{\partial\ttp^*}{\partial\xe}-1,\quad\frac{\partial\qx}{\partial\ye}=\frac{\partial\qx}{\partial\ttp^*}\cdot\frac{\partial\ttp^*}{\partial\ye},\\
            \frac{\partial\qx}{\partial\vpx}&=\frac{\partial\qx}{\partial\ttp^*}\cdot\frac{\partial\ttp^*}{\partial\vpx}-R_1\cos\ttp^*\cdot\tpt(\ttp^*),\\
            \frac{\partial\qx}{\partial\vpy}&=\frac{\partial\qx}{\partial\ttp^*}\cdot\frac{\partial\ttp^*}{\partial\vpy}-R_2\cos\ttp^*\cdot\tpt(\ttp^*),
        \end{aligned}
    \label{qxpr}
    \end{equation}
    and for $\qy$ we have
    \begin{equation}
        \begin{aligned}
            \frac{\partial\qy}{\partial\xp}&=\frac{\partial\qy}{\partial\ttp^*}\cdot\frac{\partial\ttp^*}{\partial\xp},\quad\frac{\partial\qy}{\partial\yp}=\frac{\partial\qy}{\partial\ttp^*}\cdot\frac{\partial\ttp^*}{\partial\yp}+1,\\
            \frac{\partial\qy}{\partial\xe}&=\frac{\partial\qy}{\partial\ttp^*}\cdot\frac{\partial\ttp^*}{\partial\xe},\quad\frac{\partial\qy}{\partial\ye}=\frac{\partial\qy}{\partial\ttp^*}\cdot\frac{\partial\ttp^*}{\partial\ye}-1,\\
            \frac{\partial\qy}{\partial\vpx}&=\frac{\partial\qy}{\partial\ttp^*}\cdot\frac{\partial\ttp^*}{\partial\vpx}-R_1\sin\ttp^*\cdot\tpt(\ttp^*),\\
            \frac{\partial\qy}{\partial\vpy}&=\frac{\partial\qy}{\partial\ttp^*}\cdot\frac{\partial\ttp^*}{\partial\vpy}-R_2\sin\ttp^*\cdot\tpt(\ttp^*),
        \end{aligned}
    \label{qypr}
    \end{equation}
    where
    \begin{equation*}
        \begin{aligned}
            &R_1=\frac{-\vpx\sin^2\ttp^*+\vpy\sin\ttp^*\cos\ttp^*}{\sqrt{\vpb^2-(\vpx\sin\ttp^*-\vpy\cos\ttp^*)^2}}-\cos\ttp^*,\\&R_2=\frac{-\vpy\cos^2\ttp^*+\vpx\sin\ttp^*\cos\ttp^*}{\sqrt{\vpb^2-(\vpx\sin\ttp^*-\vpy\cos\ttp^*)^2}}-\sin\ttp^*.
        \end{aligned}
    \end{equation*}
    Notice that
    \begin{equation}\label{eq:R1R2relation}
        R_1\cos\ttp^*+R_2\sin\ttp^*=-1.
    \end{equation}
    Moreover, since $t_f$ and $\ttp^*$ are the optimal value and optimal solution of~\eqref{oppr}, respectively, we know that $\frac{\partial t}{\partial\ttp}=0$ in \eqref{tpm} at $\ttp=\ttp^*$ under strategies \eqref{opst2}, i.e.,
    \begin{multline}
            \frac{1}{\vpb^2-\veb^2}(\frac{h(\bx,\ttp^*)}{g(\bx,\ttp^*)}-1)(\px\frac{\partial\qx}{\partial\ttp^*}+\qx\frac{\partial\px}{\partial\ttp^*}+\py\frac{\partial\qy}{\partial\ttp^*}\\+\qy\frac{\partial\py}{\partial\ttp^*})-\frac{1}{g(\bx,\ttp^*)}(\qx\frac{\partial\qx}{\partial\ttp^*}+\qy\frac{\partial\qy}{\partial\ttp^*})=0.
    \label{part0}
    \end{multline}
    We next compute partial derivative of \eqref{tpm} with respect to $\xp$ using \eqref{pxpr}-\eqref{part0} and we obtain
    \begin{equation}
        \begin{aligned}
            &\frac{\partial V}{\partial\xp}\\
            =&\frac{1}{\vpb^2-\veb^2}(\frac{h(\bx,\ttp^*)}{g(\bx,\ttp^*)}-1)(\px(\frac{\partial\qx}{\partial\ttp^*}\cdot\frac{\partial\ttp^*}{\partial\xp}+1)\\
            &\qquad+\qx\frac{\partial\px}{\partial\ttp^*}\cdot\frac{\partial\ttp^*}{\partial\xp}+\py\frac{\partial\qy}{\partial\ttp^*}\cdot\frac{\partial\ttp^*}{\partial\xp}+\qy\frac{\partial\py}{\partial\ttp^*}\cdot\frac{\partial\ttp^*}{\partial\xp})\\&-\frac{1}{g(\bx,\ttp^*)}(\qx(\frac{\partial\qx}{\partial\ttp^*}\cdot\frac{\partial\ttp^*}{\partial\xp}+1)+\qy\frac{\partial\qy}{\partial\ttp^*}\cdot\frac{\partial\ttp^*}{\partial\xp})\\
            =&(\frac{1}{\vpb^2-\veb^2}(\frac{h(\bx,\ttp^*)}{g(\bx,\ttp^*)}-1)(\px\frac{\partial\qx}{\partial\ttp^*}+\qx\frac{\partial\px}{\partial\ttp^*}+\py\frac{\partial\qy}{\partial\ttp^*}\\
            &\qquad+\qy\frac{\partial\py}{\partial\ttp^*})-\frac{1}{g(\bx,\ttp^*)}(\qx\frac{\partial\qx}{\partial\ttp^*}+\qy\frac{\partial\qy}{\partial\ttp^*}))\frac{\partial\ttp^*}{\partial\xp}\\
            &+(\frac{h(\bx,\ttp^*)}{g(\bx,\ttp^*)}-1)\frac{\px}{\vpb^2-\veb^2}-\frac{\qx}{g(\bx,\ttp^*)}\\=
            &(\frac{h(\bx,\ttp^*)}{g(\bx,\ttp^*)}-1)\frac{\px}{\vpb^2-\veb^2}-\frac{\qx}{g(\bx,\ttp^*)}.
        \end{aligned}
        \label{v/px}
    \end{equation}
    Similarly, we have the following
    \begin{equation}
        \begin{aligned}
            \frac{\partial V}{\partial\yp}=&(\frac{h(\bx,\ttp^*)}{g(\bx,\ttp^*)}-1)\frac{\py}{\vpb^2-\veb^2}-\frac{\qy}{g(\bx,\ttp^*)},\\
            \frac{\partial V}{\partial\xe}=&-(\frac{h(\bx,\ttp^*)}{g(\bx,\ttp^*)}-1)\frac{\px}{\vpb^2-\veb^2}+\frac{\qx}{g(\bx,\ttp^*)},\\
            \frac{\partial V}{\partial\ye}=&-(\frac{h(\bx,\ttp^*)}{g(\bx,\ttp^*)}-1)\frac{\py}{\vpb^2-\veb^2}+\frac{\qy}{g(\bx,\ttp^*)},\\
            \frac{\partial V}{\partial\vpx}=&\frac{1}{\vpb^2-\veb^2}(\frac{h(\bx,\ttp^*)}{g(\bx,\ttp^*)}-1)\cdot(-\px R_1\cos\ttp^*\cdot\tpt(\ttp^*)\\
            &\qquad-\py R_1\sin\ttp^*\cdot\tpt(\ttp^*)+\qx(1+R_1\cos\ttp^*)\\&\qquad+\qy R_1\sin\ttp^*)\\&+\frac{\qx R_1\cos\ttp^*\cdot\tpt(\ttp^*)+\qy R_1\sin\ttp^*\cdot\tpt(\ttp^*)}{g(\bx,\ttp^*)},\\
            \frac{\partial V}{\partial\vpy}=&\frac{1}{\vpb^2-\veb^2}(\frac{h(\bx,\ttp^*)}{g(\bx,\ttp^*)}-1)\cdot(-\px R_2\cos\ttp^*\cdot\tpt(\ttp^*)\\&\qquad-\py R_2\sin\ttp^*\cdot\tpt(\ttp^*)+\qy(1+R_2\sin\ttp^*)\\&\qquad+\qx R_2\cos\ttp^*)\\&+\frac{\qx R_2\cos\ttp^*\cdot\tpt(\ttp^*)+\qy R_2\sin\ttp^*\cdot\tpt(\ttp^*)}{g(\bx,\ttp^*)}.
        \end{aligned}
    \label{vpar}
    \end{equation}
    Substituting \eqref{v/px} and \eqref{vpar} into the HJI equation~\eqref{hji}, we obtain
    \begin{equation}
        \begin{aligned}
            &\frac{\partial V}{\partial\xp}\vpx+\frac{\partial V}{\partial\yp}\vpy+\frac{\partial V}{\partial\xe}\ve^*\cos\tte^*+\frac{\partial V}{\partial\ye}\ve^*\sin\tte^*\\
            &+\frac{\partial V}{\partial\vpx}\ap^*\cos\ttp^*+\frac{\partial V}{\partial\vpy}\ap^*\sin\ttp^*+1\\
            =&\frac{1}{\vpb^2-\veb^2}(\frac{h(\bx,\ttp^*)}{g(\bx,\ttp^*)}-1)\\
            &\quad\cdot(\px(\vpx-\veb\cos\tte^*+\apb\cos\ttp^*\cdot\tpt(\ttp^*))\\
            &\qquad+\py(\vpy-\veb\sin\tte^*+\apb\sin\ttp^*\cdot\tpt(\ttp^*)))\\
            &-\frac{\qx(\vpx-\veb\cos\tte^*+\apb\cos\ttp^*\cdot\tpt(\ttp^*))}{g(\bx,\ttp^*)}\\
            &-\frac{\qy(\vpy-\veb\sin\tte^*+\apb\sin\ttp^*\cdot\tpt(\ttp^*))}{g(\bx,\ttp^*)}+1\\
            =&\frac{1}{\vpb^2-\veb^2}(\frac{h(\bx,\ttp^*)}{g(\bx,\ttp^*)}-1)(\vpb^2-\px\veb\cos\tte^*-\py\veb\sin\tte^*)\\
            &-\frac{h(\bx,\ttp^*)-\qx\veb\cos\tte^*-\qy\veb\sin\tte^*}{g(\bx,\ttp^*)}+1\\
            =&\frac{1}{\vpb^2-\veb^2}(\frac{h(\bx,\ttp^*)}{g(\bx,\ttp^*)}-1)(\veb^2-\px\veb\cos\tte^*-\py\veb\sin\tte^*)\\
            &+\frac{\qx\veb\cos\tte^*+\qy\veb\sin\tte^*}{g(\bx,\ttp^*)},
        \end{aligned}
    \label{yzhji}
    \end{equation}
    where we used~\eqref{eq:R1R2relation} in the first equality, and~\eqref{pq},~\eqref{eq:relationofp}~and~\eqref{gh} in the second equality. From \eqref{xff}, we have
    \begin{equation}
        \begin{aligned}
            x_f&=\xp+\vpx t_f+\frac{1}{2}\apb\cos\ttp^*\cdot(2t_f\tpt(\ttp^*)-\tpt^2(\ttp^*))\\
            &=\px t_f+\qx+\xe,\\
            y_f&=\yp+\vpy t_f+\frac{1}{2}\apb\sin\ttp^*\cdot(2t_f\tpt(\ttp^*)-\tpt^2(\ttp^*))\\
            &=\py t_f+\qy+\ye.
        \end{aligned}
    \label{xfff}
    \end{equation}
    Then substituting~\eqref{t2t}~and~\eqref{xfff} into \eqref{opst2}, we have
    \begin{equation}
        \veb\cos\tte^*=\px+\frac{\qx}{t_f},\quad\veb\sin\tte^*=\py+\frac{\qy}{t_f}.
    \label{huajian}
    \end{equation}
    Finally substituting \eqref{tpm} and \eqref{huajian} into \eqref{yzhji}, we have
    \begin{equation*}
        \begin{aligned}
            &\frac{1}{\vpb^2-\veb^2}(\frac{h(\bx,\ttp^*)}{g(\bx,\ttp^*)}-1)(\veb^2-\px\veb\cos\tte^*-\py\veb\sin\tte^*)\\
            &+\frac{\qx\veb\cos\tte^*+\qy\veb\sin\tte^*}{g(\bx,\ttp^*)}\\
            =&\frac{1}{\vpb^2-\veb^2}(\frac{h(\bx,\ttp^*)}{g(\bx,\ttp^*)}-1)(\veb^2-\frac{h(\bx,\ttp^*)}{t_f}-\vpb^2)\\
            &+\frac{\qx\veb\cos\tte^*+\qy\veb\sin\tte^*}{g(\bx,\ttp^*)}\\
            =&(\frac{h(\bx,\ttp^*)-g(\bx,\ttp^*)}{g(\bx,\ttp^*)})(-\frac{g(\bx,\ttp^*)}{g(\bx,\ttp^*)-h(\bx,\ttp^*)})\\
            &+\frac{\qx\veb\cos\tte^*+\qy\veb\sin\tte^*}{g(\bx,\ttp^*)}\\
            =&1+\frac{\qx\veb\cos\tte^*+\qy\veb\sin\tte^*}{g(\bx,\ttp^*)}\\
            =&\frac{1}{g(\bx,\ttp^*)}(g(\bx,\ttp^*)+h(\bx,\ttp^*)+\frac{\qx^2+\qy^2}{t_f})\\
            =&\frac{1}{g(\bx,\ttp^*)}(g(\bx,\ttp^*)+h(\bx,\ttp^*)+\frac{(\vpb^2-\veb^2)(\qx^2+\qy^2)}{g(\bx,\ttp^*)-h(\bx,\ttp^*)})\\
            =&0,
        \end{aligned}
    \end{equation*}
    where we used \eqref{eq:relationofp} in the first equality, and~\eqref{tpm} in the second and penultimate equality. Thus the value function $V=t_f$ satisfies HJI equation~\eqref{hji}, which means the strategies in \eqref{opst2} are optimal in the sense of Nash equilibrium.

    Lastly, we demonstrate the continuity of the value function~\eqref{valuef} when the strategies switch. The boundary between the two strategies is when $t_f=\tpt(\ttp^*)$, i.e., the capture occurs precisely when $P$ reaches its maximum speed. We aim to show that applying the strategies in \eqref{opst2} yields a capture time $t_f$ and acceleration direction $\ttp^*$ such that $t_f = \tpt(\ttp^*)$ if and only if applying the strategies in \eqref{ops1} results in the same capture time $t_f$ and acceleration direction $\ttp^*$, thereby also satisfying $t_f = \tpt(\ttp^*)$, which means when the solution $t_f=\tpt(\ttp^*)$ satisfies~\eqref{t2t}, then~\eqref{t2t} is equivalent to~\eqref{queq}. We substitute $t_f=\tpt(\ttp^*)$ into \eqref{t2t} and obtain
    \begin{equation}
        \begin{aligned}
            &(\vpb^2-\veb^2)\tpt^2(\ttp^*)+2(\px\qx+\py\qy)\tpt(\ttp^*)+\qx^2+\qy^2=0\\
            \Leftrightarrow&(\vpb^2-\veb^2)\tpt^2(\ttp^*)+2(\px\qx+\py\qy)\tpt(\ttp^*)\\&+\qx^2+\qy^2+\apb\tpt^2(\ttp^*)(\veb\tpt(\ttp^*)\\&\qquad-\veb\cos^2\ttp^*\cdot\tpt(\ttp^*)-\veb\sin^2\ttp^*\cdot\tpt(\ttp^*))=0\\
            \Leftrightarrow&(\qx+\px\tpt(\ttp^*))^2+(\qy+\py\tpt(\ttp^*))^2-\veb^2\tpt^2(\ttp^*)\\&+\apb\tpt^2(\ttp^*)(\veb\tpt(\ttp^*)-(\qx+\px\tpt(\ttp^*))\cos\ttp^*\\&\qquad-(\qy+\py\tpt(\ttp^*))\sin\ttp^*)=0\\
            \lra&(\qx+\px\tpt(\ttp^*))^2-\apb\tpt^2(\ttp^*)\cos\ttp^*\cdot(\qx+\px\tpt(\ttp^*))\\&+(\qy+\py\tpt(\ttp^*))^2-\apb\tpt^2(\ttp^*)\sin\ttp^*\cdot(\qy+\py\tpt(\ttp^*))\\&+\apb\veb\tpt^3(\ttp^*)-\veb^2\tpt^2(\ttp^*)=0\\
            \lra&(\qx+\px\tpt(\ttp^*))^2-\apb\tpt^2(\ttp^*)\cos\ttp^*\cdot(\qx+\px\tpt(\ttp^*))\\&+\frac{1}{4}\apb^2\tpt^4(\ttp^*)\cos^2\ttp^*+(\qy+\py\tpt(\ttp^*))^2\\&-\apb\tpt^2(\ttp^*)\sin\ttp^*\cdot(\qy+\py\tpt(\ttp^*))+\frac{1}{4}\apb^2\tpt^4(\ttp^*)\sin^2\ttp^*\\&-\frac{1}{4}\apb^2\tpt^4(\ttp^*)+\apb\veb\tpt^3(\ttp^*)-\veb^2\tpt^2(\ttp^*)=0\\
            \Leftrightarrow&(\qx+\px\tpt(\ttp^*)-\frac{1}{2}\apb\cos\ttp^*\cdot\tpt^2(\ttp^*))^2+(\qy+\py\tpt(\ttp^*)\\&-\frac{1}{2}\apb\sin\ttp^*\cdot\tpt^2(\ttp^*))^2-(\frac{1}{2}\apb\tpt^2(\ttp^*)-\veb\tpt(\ttp^*))^2=0\\
            \Leftrightarrow&(\xp-\xe+\vpx\tpt(\ttp^*))^2+(\yp-\ye+\vpy\tpt(\ttp^*))^2\\&-(\frac{1}{2}\apb\tpt^2(\ttp^*)-\veb\tpt(\ttp^*))^2=0,
        \end{aligned}
    \end{equation}
    where we used~\eqref{pq}, \eqref{eq:relationofp}, \eqref{huajian}, as well as the property that $\ttp^*$ and $\tte^*$ are equal when $t_f = \tpt(\ttp^*)$. Thus we arrive at~\eqref{queq} with $t_f=\tpt(\ttp^*)$, which means the value function~\eqref{valuef} is continuous when the strategies change.

\section*{References}

\bibliographystyle{unsrt}
\bibliography{refs}

\end{document}